\documentclass{article} 
\usepackage{preprint_style,times}
\usepackage{dsfont}
\usepackage{algorithm}
\usepackage{algpseudocode}  
\usepackage{wrapfig}
\usepackage{float}          
\usepackage[utf8]{inputenc} 
\usepackage[T1]{fontenc}    
\usepackage{hyperref}       
\usepackage{url}            
\usepackage{booktabs}       
\usepackage{enumitem}

\usepackage{amsfonts}       
\usepackage{nicefrac}       
\usepackage{microtype}      
\usepackage[table]{xcolor}         
\usepackage{xspace}
\usepackage{lineno}

\usepackage{mathtools}
\usepackage{todonotes}
\usepackage{amssymb,fge}
\usepackage{thm-restate}
\usepackage{wrapfig}
\usepackage{bbm}
\usepackage{mathrsfs}
\usepackage{soul}
\usepackage{array}
\usepackage{multirow}
\usepackage{algorithm}
\usepackage[commentColor=black,beginLComment=/*~, endLComment=~*/]{algpseudocodex}
\usepackage{subcaption}
\usepackage[font=small]{caption}
\usepackage{pifont}
\usepackage{tikz-3dplot}
\usepackage{pgfplots}
\pgfplotsset{compat=newest}
\usepackage{tikzscale}
\usepackage{comment} 
\PassOptionsToPackage{final,pagebackref=true}{hyperref}
\usepackage{hyperref} 
\usepackage{cleveref}
\usepackage{wasysym}
\usepackage{acronym}
\usepackage{xfrac}
\usepackage{multibib}
\usepackage{float}

\newcommand*{\backrefalt}[4]{%
    \ifcase #1 \footnotesize{(Not cited.)}%
    \or        \footnotesize{(Cited on page~#2)}%
    \else      \footnotesize{(Cited on pages~#2)}%
    \fi}
\newcolumntype{P}[1]{>{\centering\arraybackslash}p{#1}}

\hypersetup{
	colorlinks=true,       
	linkcolor=blue,        
	citecolor=blue,        
	filecolor=magenta,     
	urlcolor=blue         
}

\usepackage{amsmath, amsthm, amsfonts, bm}
\usepackage{mathtools}
\usepackage{thmtools}
\usepackage{thm-restate}
\usepackage{xcolor}
\usepackage{xspace}
\usepackage[framemethod=TikZ]{mdframed}

\crefname{nfmdefinition}{Definition}{Definitions}

\def\1{\bm{1}}

\DeclareMathAlphabet{\mathsfit}{\encodingdefault}{\sfdefault}{m}{sl}
\SetMathAlphabet{\mathsfit}{bold}{\encodingdefault}{\sfdefault}{bx}{n}

\def\gA{{\mathcal{A}}}
\def\gB{{\mathcal{B}}}

\def\gD{{\mathcal{D}}}

\def\gF{{\mathcal{F}}}

\def\gI{{\mathcal{I}}}

\def\gR{{\mathcal{R}}}

\def\gV{{\mathcal{V}}}

\def\gX{{\mathcal{X}}}

\newcommand{\R}{\mathbb{R}}

\newcommand{\xhdr}[1]{{\noindent\bfseries #1}.}
\newcommand{\cut}[1]{}

\newcommand{\namelong}{\textsc{Systemic Adversarial Robustness}}
\newcommand{\nameshort}{\textsc{SAR}\xspace}

\mdfdefinestyle{MyFrame}{%
    linecolor=black,
    outerlinewidth=.3pt,
    roundcorner=5pt,
    innertopmargin=1pt,
    innerbottommargin=1pt,
    innerrightmargin=1pt,
    innerleftmargin=1pt,
    backgroundcolor=black!0!white}

\mdfdefinestyle{MyFrame2}{%
    linecolor=white,
    outerlinewidth=1pt,
    roundcorner=1pt,
    innertopmargin=3pt,
    innerbottommargin=2pt,
    innerrightmargin=7pt,
    innerleftmargin=7pt,
    backgroundcolor=black!3!white}

\mdfdefinestyle{MyFrameEq}{%
    linecolor=white,
    outerlinewidth=0pt,
    roundcorner=0pt,
    innertopmargin=0pt,
    innerbottommargin=0pt,
    innerrightmargin=7pt,
    innerleftmargin=7pt,
    backgroundcolor=black!3!white}

\mdfdefinestyle{nfmdefinitionstyle}{
    backgroundcolor=black!3,
    linecolor=black!35,
    linewidth=0.5pt,
    roundcorner=0pt,
    innerleftmargin=9pt,
    innerrightmargin=9pt,
    innertopmargin=6pt,
    innerbottommargin=6pt,
    skipabove=10pt,
    skipbelow=10pt
}

\theoremstyle{definition}
\newmdtheoremenv[
    style=nfmdefinitionstyle
]{nfmdefinition}{Definition}
\theoremstyle{plain}

\title{Exposing Vulnerability Monoculture Through Systemic Adversarial Robustness}

\author{%
Leo Elmecker-Plakolm\textsuperscript{1}, Lukasz Szpruch\textsuperscript{2,3}, Avishek Joey Bose\textsuperscript{1,4}\thanks{Equal advising}, Matthew Wicker\textsuperscript{1}\footnotemark[1] \\
\normalfont\textsuperscript{1}Imperial College London, \textsuperscript{2}The University of Edinburgh, \\ \textsuperscript{3}The Alan Turing Institute, \textsuperscript{4}Mila - Québec AI Institute \\[4pt]
}

\iclrfinalcopy 
\begin{document}

\maketitle
{\renewcommand{\thefootnote}{}\footnotetext{Correspondence: \texttt{\{leo.elmecker-plakolm24, joey.bose, mwicker\}@imperial.ac.uk}}
\begin{abstract}

\looseness=-1

\looseness=-1
Beneath the apparent diversity of foundation models lies an increasingly homogeneous technical stack of common training pipelines, distillation methods, and data. 
These shared algorithmic choices can expose an otherwise distinct model ecosystem to concentrated risks leading to a \emph{vulnerability monoculture}. Yet current practice in assessing adversarial robustness is restricted to individual robustness scores and attack transfer rates, leaving open the unexplored fundamental question: how many attacks suffice for catastrophic widespread compromise? In this paper, we investigate the empirical nature of shared vulnerabilities within a model pool and introduce \emph{Systemic Adversarial Robustness} (\nameshort), a novel measure of shared vulnerability of a model pool through the attack diversity needed to exploit it. Specifically, \nameshort interpolates between single-attack transferability and independent per-model robustness: small values reveal a vulnerability monoculture, whereas large values indicate that failures are dispersed across models. We theoretically formulate \nameshort as a partial set-cover problem with provable approximation guarantees---enabling the development of scalable algorithms to empirically measure the systemic adversarial robustness. We empirically test the shared vulnerability of vision and large language models and characterize their shared collective vulnerability curve. Specifically, we demonstrate that model pools distilled from DeepSeek-R1 and Claude Opus 4.6 exhibit lower estimated systemic robustness than their corresponding parent pools, revealing a potential systemic cost of frontier-model distillation. Our framework seeks to establish vulnerability monoculture as a measurable target for a more wholistic adversarial evaluation of model populations. 

\end{abstract}

\section{Introduction}
\label{sec:introduction}

\looseness=-1
Diversity is one of the oldest defenses against systemic failure. In nature, genetic diversity protects populations against extinction from deadly viruses, and a computer network running identical software can fall prey to a single exploit~\citep{geer2003cyberinsecurity}. As a result, the danger of fostering a monoculture lies not in any individual member being weak, but rather in all members in the population sharing the same weakness. This intuition applies in several facets of real-world life, but also in algorithmic decision-making as \emph{algorithmic monoculture}~\citep{algorithmicmonoculture}, which underscores the paradox: when many decision-makers converge on the same algorithm, the collective can be worse off even if the algorithm is better for each of them in isolation. The same effects of algorithmic monoculture are now, unfortunately, measurable in machine learning through the lens of \emph{outcome homogenization}, wherein foundation models that are trained on common data and methodology make correlated errors in an ecosystem~\citep{li2025quantifyingcorrelationsmachinelearning,bommasani2021opportunities,toups2023ecosystem}.

\looseness=-1
As foundation models grow stronger in capabilities, the risk and potential damage of correlated errors are amplified~\citep{goel2025great,kim2025correlated}. This growth in model capabilities is also coupled with the Cambrian explosion of new applications for foundation models and agentic systems that touch a diversity of economic sectors, including safety-critical domains such as healthcare, autonomous driving, and finance~\citep{bommasani2021opportunities, szpruch2025insuring, yutysever2020surveryofautonomousdriving}. However, a diverse foundation model ecosystem need not have diverse vulnerabilities. Powering this vast foundation model zoo is a deceptively narrow technical stack of overlapping pre-training data, near-identical architectural classes, and shared post-training recipes, and, more recently, distillation methods from which student models inherit capabilities. This convergence to homogenization runs the risk of breeding a new \emph{vulnerability monoculture} where the exploitable defects of larger foundation models are silently propagated to every student model distilled from it. 

\looseness=-1
Indeed, the widespread adoption of distillation as a standard stage in post-training only serves to accelerate the convergence of shared vulnerabilities. For instance, open models are routinely fine-tuned on the outputs of a few frontier teachers~\citep{xu2024survey}. For example, DeepSeek-R1~\citep{guo2025deepseek} was distilled by its developers into Qwen~\citep{bai2023qwen} and Llama models~\citep{grattafiori2024llama}, and community checkpoints of open models from five developers have been post-trained on outputs attributed to Claude Opus~4.5. This surface area for shared vulnerabilities remains opaque, as distillation can transmit traits that are invisible in the training data itself~\citep{cloud2025subliminal}. 

\looseness=-1
Despite the increasing risk of shared vulnerabilities, current adversarial robustness evaluations fundamentally cannot distinguish a vulnerability monoculture from a genuinely diverse pool. In particular, individual attack success rates~\citep{carlini2017towards}
assume adversaries that attack every model separately, and transferability assumes adversaries who attack with a single input. Importantly, real-world adversaries operate in between and may seek to maximize the attack surface exposed through shared vulnerabilities. At present, the quantity that governs this threat model lacks a prescriptive, systematic measure: \emph{how many distinct attacks does it take to compromise a pool of models?}

\looseness=-1
\xhdr{Present work} In this work, we investigate an adversarial manifestation of the vulnerability monoculture in foundation models and seek to quantify the shared vulnerabilities.
To expose this, we introduce and define the notion of \emph{\namelong} (\nameshort) that aims to measure a new notion of model diversity, not through how different models are constructed but rather in how \emph{differently} models can be broken. Specifically, \nameshort measures the minimum number of distinct attacks, within a budget, that together compromise at least a $1-\delta$ fraction of the model pool. Critically, sweeping $\delta$ traces the model pool's full vulnerability profile, which complements any existing adversarial robustness measure. Interpreting this vulnerability profile also provides a signature wherein a low \nameshort corroborates high shared risk, as depicted in~\cref{fig:titlefig} for distillation of foundation models.

\looseness=-1
We summarize our contributions as follows:
\begin{enumerate}[topsep=0pt, partopsep=0pt, itemsep=0pt, parsep=0pt, leftmargin=*]
    \item \looseness=-1 \xhdr{Systemic adversarial robustness} We introduce and define the concept of \textit{systemic adversarial robustness} as a new form of risk assessment across an entire pool of models. Specifically, we aim to find the smallest number of perturbation vectors that manage to \textit{fool} a given proportion of models in a fixed (distilled) model pool (\S\ref{sec:systemic_adversarial_robustness}).
    \item \looseness=-1 \xhdr{Algorithms} We formulate \nameshort as a partial set-cover problem, and give a top-down greedy estimator in~\cref{alg:top_down} that requires only a joint attack oracle (\S\ref{sec:algorithms}).
    With an $\alpha$-approximate oracle, the estimate lies within a factor $1+\alpha^{-1}\ln\lceil(1-\delta)M\rceil$ of the true value, a logarithmic dependence that cannot be improved in general. Importantly, the procedure attains a $(1-e^{-\alpha})$ fraction of the best coverage achievable with $k$ attacks. 
    \item \looseness=-1 \xhdr{Empirical analysis} We measure \nameshort on a DeepSeek-R1 distilled model pool and community checkpoints post-trained on outputs attributed to Claude Opus~4.5 (\S\ref{sec:experiments}). Our investigation points to increasing evidence of an emerging vulnerability monoculture, as shown by measuring \nameshort.

\end{enumerate}

\begin{figure*}
    \centering
    \includegraphics[width=\textwidth]{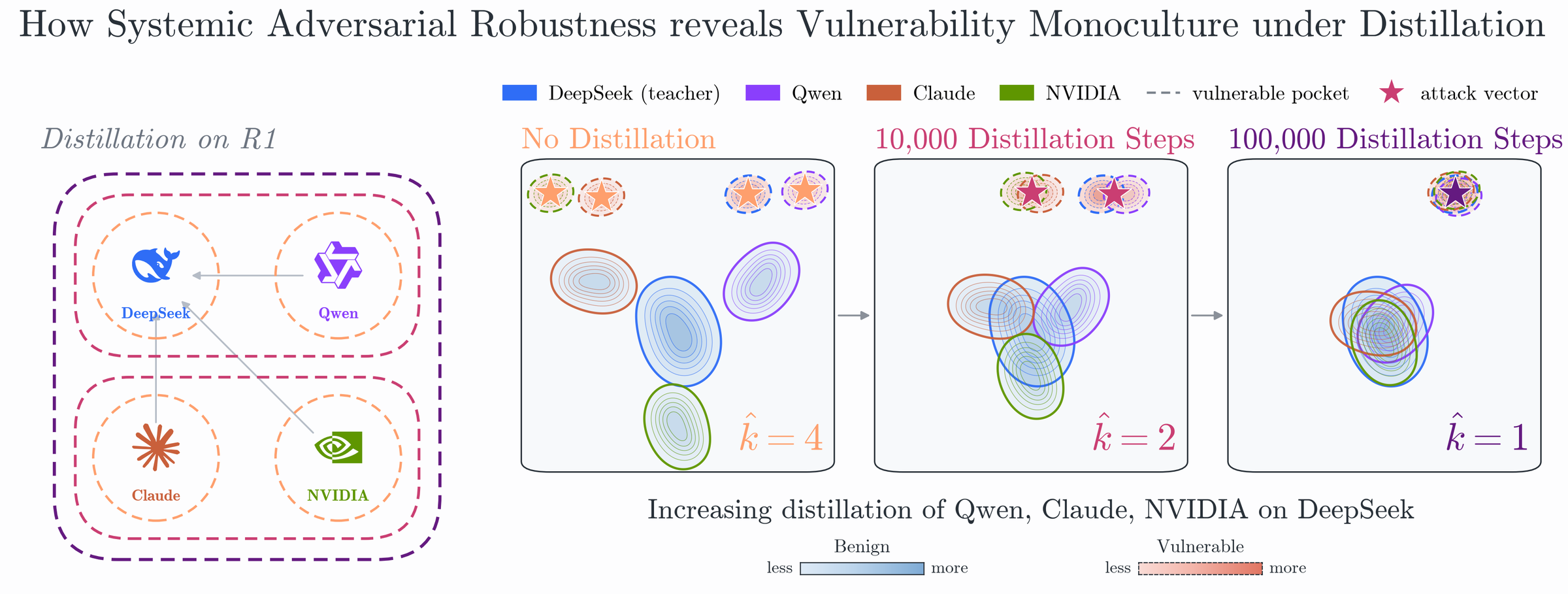}
    \vspace{-15pt}
    \caption{\small \looseness=-1 \textbf{Left:} A stylized visualization of how a decrease in {\namelong} emerges through model distillation in large language models. \textbf{Right} Increasing the number of distillation fine-tuning steps reduces the number of unique attack vectors $\hat{k}$ needed to compromise a model pool due to the internalization of shared vulnerabilities, visualized as the overlap of vulnerable level sets.}\label{fig:titlefig}
    \vspace{-20pt}
\end{figure*}

\section{Preliminaries}
\label{sec:prelim}


\looseness=-1
\xhdr{Notation and convention}
Let $\mathcal D:= \{(x_i, y_i)\}^{N}_{i=1}$ be a set of data points sampled from $\mathcal{P}(\mathcal X, \mathcal Y)$ where $\gX $ is the input space 
and $\mathcal Y$ forms a task-dependent output space (e.g., classification, regression, generation). When the task requires continuous inputs, we fix $\gX := \R^d$, while in the case of discrete sequences such as in language modeling, we assume that for a sequence of length $L$ over a token vocabulary $\gV$, the input space is $\gX := \gV^L$. We describe an arbitrary function, with learnable parameters $\theta \in \Theta$, as $f^\theta: \mathcal X \rightarrow\mathcal Y$. A collection of $M$-many such functions is termed a \emph{model pool}\footnote{For generative models (e.g., probabilistic large language models) we fix a decoding rule for sample generation, e.g., greedy decoding, such that each $f^\theta$ is a deterministic function.} and is denoted as $\mathcal{F} := \{f^{\theta}_{j}\}^M_{j=1}$. To define a perturbation set, we write $\gB_{\epsilon}(x) \coloneqq \left\{
        x' \in \gX :
        d(x,x') \leq \varepsilon
    \right\}$ as the ball of ``radius`` $\varepsilon$ around the original input $x$. In the case of continuous spaces $d_{\infty}(x,x')$ is the $\ell_\infty$ distance, while for token-level attacks we set $d_{\mathrm{edit}}(x,x')$ as the Levenshtein distance, where $\varepsilon$-budget constrains the number of single-token
insertions, deletions, and substitutions required to transform
$x$ into $x'$.
Finally, an attack algorithm on the function $f_j$ is denoted $\mathcal A_{j}(\varepsilon) : \mathcal{X} \to \mathcal{X}$ which takes in an input $x$ and returns a \emph{perturbed} input $x^\star \in \gB_\epsilon$

\subsection{Adversarial Robustness}
\label{sec:background_adversarial_robustness}

\looseness=-1
We are interested in the problem of producing adversarial attacks in the permissive white-box setting, which traditionally considers the safety relationship between input-output pairs centered at an instance $(x, y)$ given oracle access to a target model $f$ or model pool $\gF$. In particular, if $S_y \subseteq \mathcal{Y}$ is the safe output region, the goal of an adversarial attack is to produce an adversarial example such that $f(\mathcal{A}(x)) \notin S_y$. This motivates
the natural measure of performance of an attack algorithm $\gA$:
\begin{equation}
\label{eq:asr}
    \mathrm{ASR} \coloneqq \frac{1}{N} \sum_{i=1}^{N}\mathds{1}\left[f^\theta_j(\mathcal{A}_j(x_{i})) \notin S_{y_{i}}\right] \quad (x_i, y_i) \sim \gD.
\end{equation}
\looseness=-1
The \emph{attack success rate} (ASR)~\citep{carlini2017towards,croce2020minimally} defined in~\cref{eq:asr} underscores the ability of an attack algorithm to break a single model. 
A different but related measure of the performance of an attack algorithm is to consider the \emph{transferability} of an attack to a new function---i.e., $f^\theta_{k}(\mathcal{A}_{j}(x)) \notin S_y$ for $k \neq j$~\citep{szegedy2014intriguing, papernot2016transferability, liu2017delving}. Against a model pool~\citep{liu2017delving,mao2022transfer}, we can equally measure the performance of an attack algorithm by computing its transfer attack success coverage:
\begin{equation}
    \label{eq:tasc}
    \texttt{tASC}(j) := \frac{1}{M} \sum_{k=1}^{M} \mathrm{AS}(j, k), \quad \text{where}\, \mathrm{AS}(j,k) \coloneqq \mathds{1}\left[f^\theta_{k}(\mathcal{A}_{j}(x)) \notin S_y\right],
\end{equation}
\looseness=-1
where $\mathrm{AS}(j,k)$ captures the transfer of an attack from models $j \to k$. In addition to the usual transfer setting, it is possible to devise an attacker that crafts a transferable attack by leveraging a \emph{cluster} of models $\mathcal{A}_{\mathcal{I}}$ where $\mathcal{I} \subseteq [M] := \{1, \dots, M\}$ is the index set of models~\citep{dong2018boosting, chen2024rethinking}.
The resulting coverage is termed the \textit{joint} attack success coverage, and is defined as $\texttt{jASC}(\mathcal{I}) := \frac{1}{M} \sum_{j=1}^{M} \mathds{1}\left[f^\theta_j\left(\mathcal{A}_{\mathcal{I}}(x)\right) \notin S_y\right]$. Lastly, we note that $\texttt{tASC}(j) = \texttt{jASC}(\{j\})$.

\section{Method}
\label{sec:method}

\looseness=-1
We are interested in analyzing the collective adversarial robustness profile of a model pool $\gF$, in which subsets of the pool $\gI \subset \gF$ are potentially distilled from a common model within the pool. More precisely, we aim to characterize whether a shared teacher distillation setting that gives rise to the model pool can inadvertently expose a new form of vulnerability inherited by the students from a common teacher. In particular, we ask whether distillation creates this tension: can otherwise distinct students acquire vulnerabilities that 
a small collection of adversarial inputs can exploit?

\looseness=-1
\xhdr{Distinction for classical adversarial robustness}
At a cursory level, determining the shared adversarial robustness profile of $\gF$ bears resemblance to conventional adversarial attacks, including both the transferable  (\texttt{tASC}) and joint attack (\texttt{jASC}) settings. However, a key distinction is that individual attack success rates measure the individual vulnerability of a model, while \texttt{tASC} and \texttt{jASC} are fundamentally about the reachability of an optimized \emph{single adversarial input} within $\gF$. Crucially, no setting assesses the shared vulnerability of the entire model pool. More concisely, widespread failure against $\gF$ may require only a few broadly effective adversarial inputs, or many inputs targeting different models. These two extremes illuminate a new attack surface for shared vulnerabilities of model clusters within $\gF$ that current training methodology, such as model distillation, blatantly expose.

\looseness=-1
We organize the remainder of the section as follows: in \S\ref{sec:systemic_adversarial_robustness} we first rigorously define a new notion of adversarial robustness, termed \namelong (\nameshort) to characterize the shared vulnerability of a model pool. We then build intuition in \S\ref{sec:warmup} by considering a warm-up exercise with attacking a model pool trained on MNIST. In \S\ref{sec:algorithms} we present a top-down greedy algorithm to compute \nameshort against practical real-world models. Finally, in \S\ref{sec:set_cover} we provide a theoretical perspective on \nameshort through the partial set cover problem, which helps codify approximation guarantees.

\subsection{Systemic Adversarial Robustness}
\label{sec:systemic_adversarial_robustness}

\looseness=-1
We now define \namelong (\nameshort) to quantify the
shared and systemic robustness profile of a prescribed model pool $\gF$. Intuitively, given a fixed model pool and perturbation set \nameshort is the minimum number of distinct inputs needed to compromise a
target fraction of the pool, normalized by its size. Stated more formally, results in the following definition below.




\begin{nfmdefinition}[$1-\delta$ Systemic Adversarial Robustness]
\label{def:sys_adv_robust}
Let $\mathcal F=\{f_j\}_{j=1}^{M}$ be a pool of $M\geq 1$
deterministic models $f_j:\mathcal X\to\mathcal Y$.
Fix an instance $(x,y)\in\mathcal X\times\mathcal Y$, an admissible
perturbation set $T_x\subseteq\mathcal X$, a safe output region
$S_y\subseteq\mathcal Y$, and $\delta\in[0,1)$.
For each admissible input $x'\in T_x$, let
\[
    C(x') := \{j\in[M] : f_j(x')\notin S_y\}
\]
denote the indices of the models it compromises, where
$[M]=\{1,\ldots,M\}$.
The \emph{systemic adversarial robustness at coverage $1-\delta$} is the
minimum number of distinct admissible inputs needed to collectively
compromise at least a $1-\delta$ fraction of the pool, normalized
by its size:
\begin{equation}
\label{eq:systemic-robustness}
    \operatorname{SAR}_{1-\delta}(\mathcal F;x,y)
    :=
    \frac{1}{M}
    \min_{V\subseteq T_x}
    \left\{
        |V|
        \;:\;
        \frac{1}{M}\left|\bigcup_{x'\in V} C(x')\right|
        \geq (1-\delta)
    \right\}.
\end{equation}
\end{nfmdefinition}
\looseness=-1
In~\cref{def:sys_adv_robust} each model counts toward coverage if it is compromised by at least
one input in $V$. If no collection satisfies the coverage constraint,
we define the systemic adversarial robustness to be $1$.
At the level of shared vulnerabilities, small \nameshort values indicate a critical shared risk profile, in that only a few crafted adversarial inputs suffice for compromising a $1-\delta$ fraction of the model pool $\gF$. Specifically, a low \nameshort implies the existence of highly \emph{correlated} vulnerabilities and thus provides a holistic robustness profile.
We also highlight that \nameshort does not provide commentary on individual robustness measurements, but rather \nameshort aids us in investigating whether
there are latent clusters---potentially due to model distillation---within $\gF$ that are vulnerable and easier to exploit collectively.

\cut{
Importantly, SAR alleviates the issues we established for joint and transferable attacks: the robustness of a pool is no longer measured by the cluster of models with the most in common or the most vulnerable cluster of models. SAR takes a wholistic look (up to a user-defined portion $\delta$) at the vulnerabilities of a pool of models. Systemic adversarial robustness being low implies highly correlated vulnerabilities. 
}


\begin{figure*}[h]
    \centering
    \includegraphics[width=\textwidth]{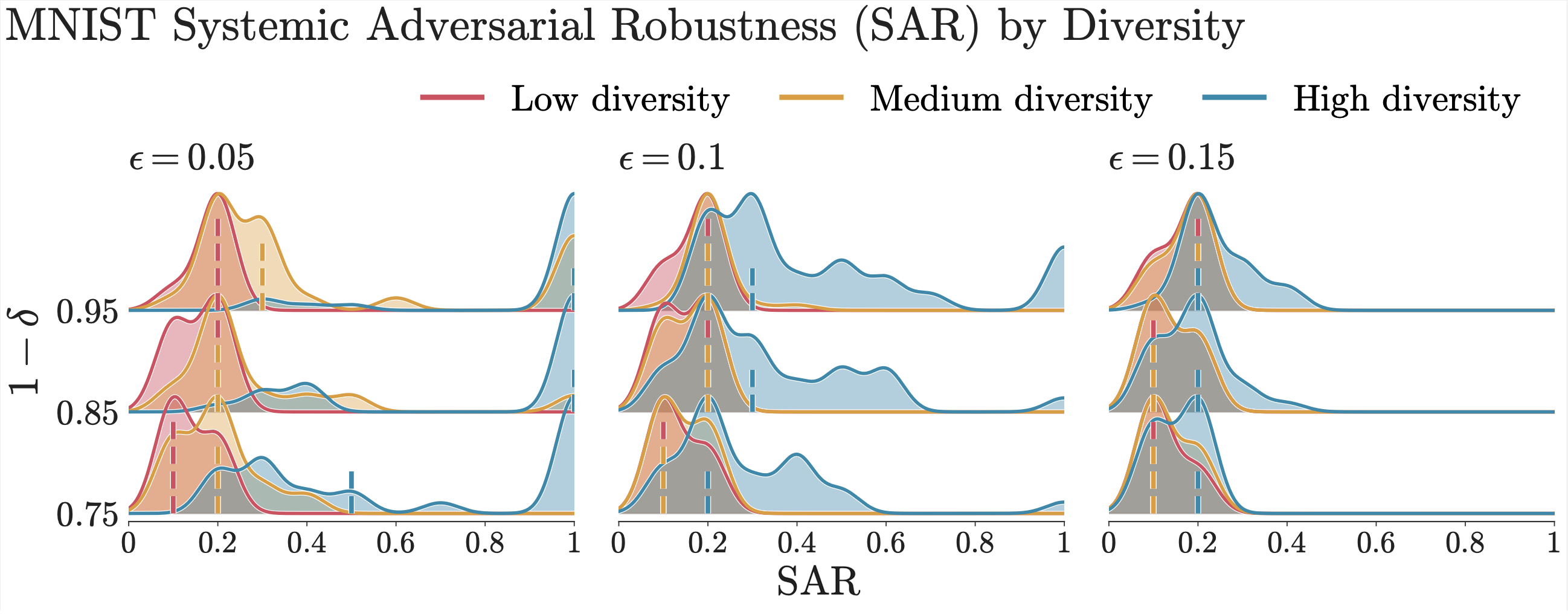}
    \vspace{-15pt}
    \caption{\small \looseness=-1 Systemic Adversarial Robustness increases as a model pool becomes more diverse.}\label{fig:mnist_diversity}
    \vspace{-15pt}
\end{figure*}

\subsection{Warmup with MNIST}
\label{sec:warmup}

\looseness=-1
We first illustrate the distinction between \nameshort and conventional adversarial robustness using a controlled experiment with MNIST classifiers. As the canonical attack algorithm $\gA$ we use the PGD attack~\citep{madry2019deeplearningmodelsresistant}. We next systematically investigate two concrete threads as a warmup:

\begin{enumerate}[label=\textbf{(Q\arabic*)},left=0pt,nosep]
\item \looseness=-1 \xhdr{Pool Diversity} How does \nameshort change as a function of diversity within the model pool $\gF$?
\item \looseness=-1 \xhdr{Distillation} How does distillation 
change the individual ASR and \nameshort?
\end{enumerate}

\looseness=-1
\xhdr{(Q1) Pool Diversity}
We construct three pools of classifiers from low to medium to high diversity. The
\emph{low-diversity} pool uses the same MLP architecture, training
data, and optimization procedure, varying $f_i \in \gF$ only in initialization.
The \emph{medium-diversity} pool additionally fixes the architecture but varies data
augmentation, training splits, optimization, and regularization. Lastly,
the \emph{high-diversity} pool also varies architecture, combining
MLPs and variants of convolutional models. In all cases, the models are trained to convergence with mean test accuracies of $98.0\%$, $98.3\%$, and $99.2\%$ (full training
details are provided in \S\ref{app:mnist}).

\looseness=-1
The first controlled experiment is visualized in~\cref{fig:mnist_diversity} and illustrates perturbations crafted within $\gB_\varepsilon$, as we vary the radius $\varepsilon$ in an $\ell_\infty$-sense. In particular, we allow each model's safe region $S_y$ to be its own clean prediction, and count the model as compromised when its prediction falls outside $S_y$, i.e.,
$    C(x') = \{j : f_j(x') \neq f_j(x)\}$. We next compute the \nameshort 
across pools, perturbation budgets, and coverage thresholds, tracing a curve as a function of $1-\delta$ along the $x$-axis. This \nameshort curve is each model pool's vulnerability profile and exposes the main finding: the high-diversity pool at the stringent $\varepsilon=0.05$ attack budget exhibits substantially greater systemic adversarial robustness than the other pools. 
As the attack strength increases by increasing the perturbation budget, we observe that the high-diversity pool sustains its \nameshort advantage.

\looseness=-1
\xhdr{(Q2) Distillation}
\looseness=-1
We next consider a controlled distillation setup to assess the impact of training on annotations of a \emph{shared} teacher. Specifically, we train a ResNet-18 teacher on MNIST ($99.5\%$ test accuracy)
and use its predictions to label $240{,}000$ EMNIST-digits images.
For this experiment, we use $M=20$ of the high-diversity model pool as students and train each on these teacher-labeled examples.

\looseness=-1
In~\cref{fig:emnist_distillation} we visualize two plots of interest. First, we plot \nameshort and ASR as a function of distillation fine-tuning steps in comparison to the reference starting point of the model pools. This experiment aims to tease apart an initial counterintuitive phenomenon of \nameshort: its independence from individual model ASR. In particular, one can observe a substantial reduction in the required distinct attack count for a specified target coverage of $1-\delta = 0.85$, without a noticeable decrease in individual attack success. This not only reveals a change in collective vulnerability due to distillation, but points to the fact that ASR is an insufficient proxy for \nameshort. 
\Cref{fig:emnist_distillation} 
illustrates the emergence of a \emph{vulnerability monoculture} as distillation exposes students in the model $\gF$ to higher shared risk. Additionally, we plot the distribution of the SAR at $1-\delta=0.85$, which clearly displays a shift in model pool robustness over the course of distillation-based fine-tuning.

\begin{figure}[h]
    \centering
    \includegraphics[width=\linewidth]{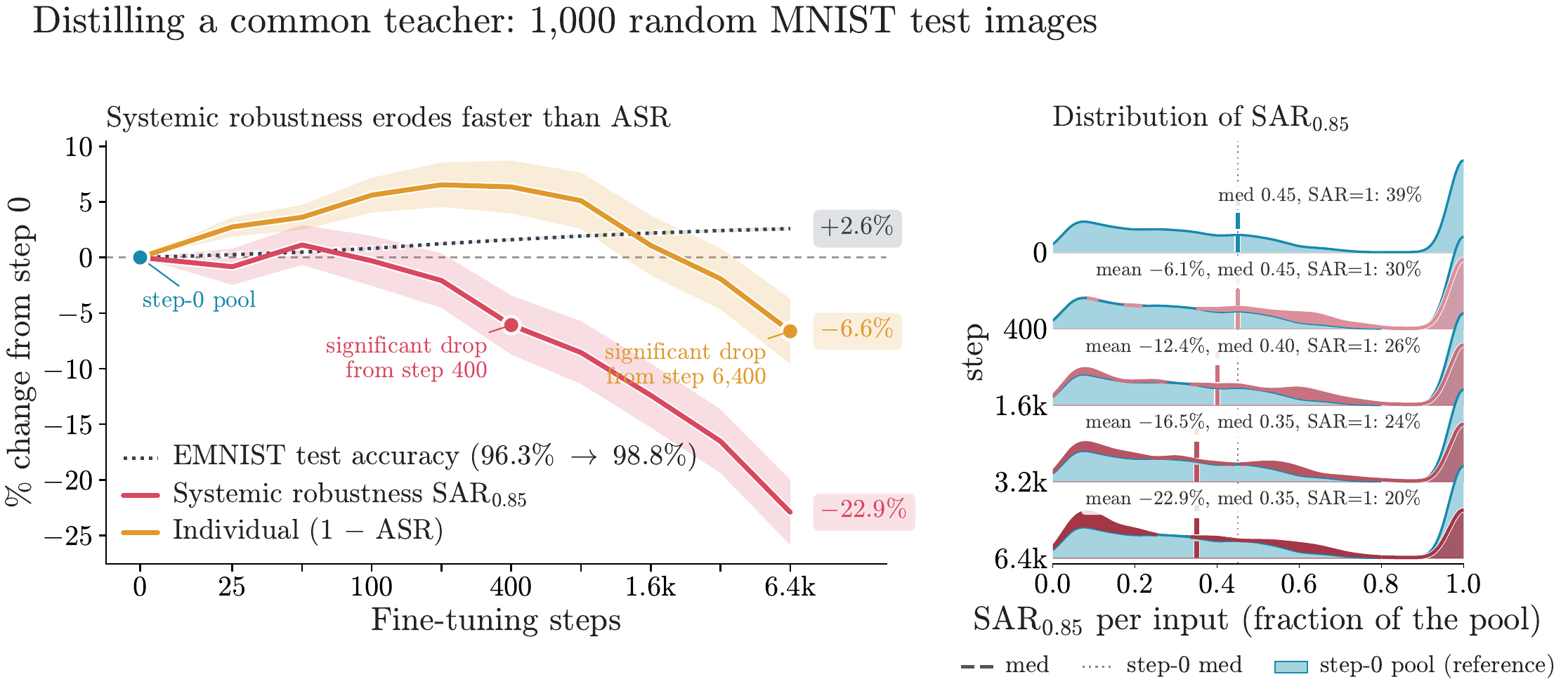}
    \vspace{-10pt}
    \caption{\small \looseness=-1 Systemic Adversarial Robustness vs Attack Success Rate under distillation. \textbf{Left:} The change in \nameshort and ASR as a function of fine-tuning steps. \textbf{Right:} Tracing the \nameshort distribution curves between the starting model pool and the fine-tuned model pool for $  \operatorname{SAR}_{1-\delta} = 0.85$ at various steps during fine-tuning. }
    \label{fig:emnist_distillation}
\end{figure}

\subsection{Estimating \namelong}
\label{sec:algorithms}

\looseness=-1
Estimating \nameshort as defined in~\cref{def:sys_adv_robust} requires finding a small collection of
admissible perturbations whose combined coverage reaches the prescribed
$1-\delta$ threshold. At the scale of LLMs and other real-world security use cases, such an exhaustive search over collections is computationally intractable. Instead, we aim to provide a computationally tractable \emph{top-down} greedy procedure that synthesizes a simple line of attack: find an input perturbation that compromises several models in $\gI \subset \gF$, record its coverage, and repeat on the models left uncovered $\gF \setminus \gI$. This top-down algorithm is outlined in~\cref{alg:top_down} and defines a recursive procedure until $\tau \coloneqq  \lceil (1-\delta) M \rceil$ is reached.

\looseness=-1
\xhdr{Discovering shared vulnerabilities}
\Cref{alg:top_down} maintains an active set of model indices
$\mathcal I$, initially $[M]$, and a covered set $\gR$, initially
empty. At each iteration, a joint attack $\mathcal A_{\mathcal I}$
searches for an admissible input with broad coverage of the
active models. If it compromises at least one member of
$\mathcal I$, we retain the input, add the newly compromised
models to $\gR$, and remove them from further targeting.
On the other hand, if $\gA_{\gI}$ compromises no active model, we attempt
a single-model attack against one member of $\gI$. This procedure continues recursively until $|\gR|\geq\tau$ or no active models remain. Specifically, if the target coverage is attained using
$\hat{k}$ distinct inputs, the algorithm returns the \nameshort estimate
$\widehat{\operatorname{SAR}}_{1-\delta}=\hat{k}/M$ which is an upper bound on the true systemic adversarial robustness of $\gF$.

\looseness=-1
It is instructive to note that the first attack launched in~\cref{alg:top_down} targets the entire model pool $|\gF|$, and every subsequent accepted attack vector is optimized against \emph{all} current models remaining in the pool $\gF \setminus \gR$. Intuitively, clusters that share the same attack surface area emerge from the observed coverage. Furthermore, a broadly transferable attack may cover much of the pool, but the additional attacks needed to reach the target coverage
determine the estimated systemic adversarial robustness.

\begin{algorithm}[t]\small
\caption{Top-Down Greedy Estimate of $1-\delta$ Systemic Robustness}
\label{alg:top_down}
\begin{algorithmic}[1]
\State \textbf{Input:} instance $(x, y)$, model pool $\mathcal F $, joint attack $\mathcal A$, threshold $\delta$
\State $\tau := \lceil (1-\delta) M \rceil$, \; $\mathcal I := [M]$, \; $\mathcal R := \emptyset$, \; $X := \emptyset$ \Comment{uncovered models, covered models, attack vectors}
\While{$|\mathcal R| < \tau$ \textbf{and} $\mathcal I \neq \emptyset$}
    \State $x' := \mathcal A_{\mathcal I}(x)$, \quad $F := \{ j \in \mathcal I : f_j(x') \notin S_y \}$ \Comment{universal vector for all uncovered models}
    \If{$F = \emptyset$} \Comment{white-box fallback}
        \State pick $j \in \mathcal I$; \; $x' := \mathcal A_{\{j\}}(x)$, \quad $F := \{ j' \in \mathcal I : f_{j'}(x') \notin S_y \}$
        \If{$F = \emptyset$} $\mathcal I := \mathcal I \setminus \{j\}$; \textbf{continue} \Comment{$j$ is unattackable under $\mathcal A$} \EndIf
    \EndIf
    \State $X := X \cup \{x'\}$, \quad $\mathcal R := \mathcal R \cup F$, \quad $\mathcal I := \mathcal I \setminus F$
\EndWhile
\State \Return $\hat k / M$ with $\hat k := |X|$ if $|\mathcal R| \geq \tau$, and $1$ otherwise \Comment{degenerate case of Definition~\ref{def:sys_adv_robust}}
\end{algorithmic}
\end{algorithm}

\subsection{Analogy to Set Covering and Approximation Guarantees}\label{sec:set_cover}

Fix an instance $(x, y)$ and associate with every perturbed input $x' \in T_x$ its \textit{cover} $\mathcal C(x') := \{ j \in [M] : f_j(x') \notin S_y \}$, i.e., the members of the pool that $x'$ breaks. Definition~\ref{def:sys_adv_robust} then reads
\begin{equation}\label{eq:psc}
    k^\star := \min\Big\{ |X| \;:\; X \subseteq T_x,\ \big| \textstyle\bigcup_{x' \in X} \mathcal C(x') \big| \geq \tau\Big\}
\end{equation}
with systemic robustness $k^\star / M$. Problem~\eqref{eq:psc} is the \textit{partial set cover} problem~\citep{kearns1994computational, slavik1997improved} over the universe $[M]$ with unit costs and the set system $\{\mathcal C(x') : x' \in T_x\}$; the choice $\delta = 0$ recovers the classical set cover problem~\citep{chvatal1979greedy}, so computing the value of Definition~\ref{def:sys_adv_robust} exactly is NP-hard in general. Two features distinguish our setting from the textbook setting. First, the set system is only implicit: $T_x$ is continuous or exponentially large, and covers are exposed one at a time by the attack. Second, the attack returns \textit{some} vector rather than the vector of maximal marginal coverage. Algorithm~\ref{alg:top_down} is therefore the greedy algorithm for partial cover run with an approximate incremental oracle~\citep{goundan2007revisiting}, whose quality we summarize in a single parameter.

 \begin{nfmdefinition}[$\alpha$-approximate joint attack]\label{def:alpha}
A joint attack $\mathcal A$ is \textit{$\alpha$-approximate}, $\alpha \in (0, 1]$, if for every non-empty $\mathcal I \subseteq [M]$
\begin{equation}
    \big| \mathcal C(\mathcal A_{\mathcal I}(x)) \cap \mathcal I \big| \;\geq\; \alpha \max_{x' \in T_x} \big| \mathcal C(x') \cap \mathcal I \big| .
\end{equation}
\end{nfmdefinition}
\looseness=-1
That is, the vector returned against the uncovered models breaks at least an $\alpha$ fraction of the number of models that the best vector in $T_x$ would break; $\alpha = 1$ is an exact greedy oracle. Under Definition~\ref{def:alpha}, and whenever \eqref{eq:psc} is feasible, the white-box fallback of Algorithm~\ref{alg:top_down} is never executed: as long as fewer than $\tau$ models are covered, some uncovered model is attackable, so the joint attack covers at least one. The fallback only guards termination when the attack underperforms.

\begin{mdframed}[style=MyFrame2]
\begin{restatable}[Approximation ratio]{theorem}{approxratio}\label{thm:approx}
\looseness=-1
Fix an instance $(x, y)$ and $\delta \in [0, 1)$ for which \eqref{eq:psc} is feasible, and let $\hat k$ be the number of vectors returned by Algorithm~\ref{alg:top_down} with an $\alpha$-approximate joint attack. Then
\begin{equation}\label{eq:ratio}
    k^\star \;\leq\; \hat k \;\leq\; \Big( 1 + \frac{\ln{\tau}}{\alpha} \Big) k^\star ,
\end{equation}
and the same ratio holds between the estimated and the true systemic robustness $\hat k / M$ and $k^\star / M$. Being uniform over instances, it carries over to the average over $\mathcal D$.
\end{restatable}
\end{mdframed}

\looseness=-1
The value of \cref{thm:approx} is not a guarantee of performance, but an intuition about when our algorithm is known to perform well. \Cref{def:sys_adv_robust} casts SAR as a partial set cover problem over the pool \citep{kearns1994computational,slavik1997improved} whose sets, the covers $C(x')$ of admissible inputs, are revealed only by running an attack. An estimate $\hat k$ from \Cref{alg:top_down} therefore carries two errors: the \emph{adversarial error} of the attack, which may miss breaking inputs that exist, and the \emph{procedural error} of the top-down procedure, which commits to covers greedily and may use more inputs than necessary. Our analysis in \cref{thm:approx} bounds only the procedural error, that is, the error the top-down procedure adds on top of the adversarial error. 

\looseness=-1
Writing $\alpha_t$ for the fraction of the best achievable marginal coverage that the attack attains at step $t$, and $\bar\alpha$ for its mean over all but the last step, every run that discards no model satisfies $k^\star\le\hat k\le 1+k^\star\ln\tau/\bar\alpha$ (\Cref{thm:approx}).  On the other hand, adversarial error itself cannot be bounded in general: for a single model, an efficient attack with any guaranteed $\alpha>0$ would decide neural-network verification, which is NP-hard \citep{katz2017reluplex}. Because $\bar\alpha$ is not observable, the theorem is a decomposition rather than a certificate, and we use it as an empirical guide.


\section{Experiments}
\label{sec:experiments}

\looseness=-1
In \S\ref{sec:warmup} we showed that models with similar individual robustness
can form pools with very different systemic robustness, and that fine-tuning on
a shared teacher can lower SAR while individual robustness holds. We now ask
whether the same happens to language models that share a post-training source.
To do so, we compare a pool of models before fine-tuning (parents) and after fine-tuning from a given teacher (descendants). We estimate
SAR with different pool-aware attacks, and control for models that already fail on
the unmodified request.

\looseness=-1
\xhdr{Model pools}
Our LLM experiments rely on three different model pools varying in model complexity and diversity. We evaluate the systemic adversarial robustness compared to the attack success rate over the course of distillation on a \textit{small} model pool with LLMs ranging from 135M to 5B parameters. Additionally, we use a \textit{medium} and \textit{large} pool to further evaluate the impact of distillation on larger and more diverse base models. The medium pool consists of five Qwen models and one Llama model. The large pool is more diverse with eight models from five different model providers. All of our experiments are distilled on \textit{DeepSeek-R1} and/or \textit{Claude Opus 4.6} traces. For further details we defer to \S\ref{app:llm_pools}.


\looseness=-1
\xhdr{Behaviors and judge}
Harmful requests are drawn from the standard category of HarmBench~\citep{mazeika2024harmbench}. Decoding is deterministic, and a model counts
as compromised when the HarmBench Llama-2-13B
classifier~\citep{mazeika2024harmbench} labels its response harmful.

\looseness=-1
\xhdr{Attacks}
We estimate SAR with the top-down procedure of \cref{sec:algorithms} and three
pool-aware joint attacks. Within the main text we consider attacks warm-started with samples from AmpleGCG~\citep{kumar2024amplegcg} and optimized using BEAST~\citep{sadasivan2024fast} as joint attack. In the Appedix we show that the same trends are found using other optimization approaches including warm starting with AdvLLM~\citep{sun2025iterative}, optimized with a joint GCG~\citep{zou2023universal}, as well as using non-insertion based attacks such as AutoDAN~\citep{liu2025autodan}.

\looseness=-1
\xhdr{Evaluation}
For joint BEAST, $\hat{k}$ is the number of attacks the top-down search selects, and behaviors whose coverage requirement is never met have $\mathrm{SAR}=1$. A model that fails on the unmodified request is compromised without any attack. We report the relative change in mean $\hat{k}/M$ from the original pool.

\subsection{Systemic Risk of Distilled Lineages}
\label{sec:teacher_fingerprints}

\begin{figure*}[t]
    \centering
    \includegraphics[width=\textwidth]{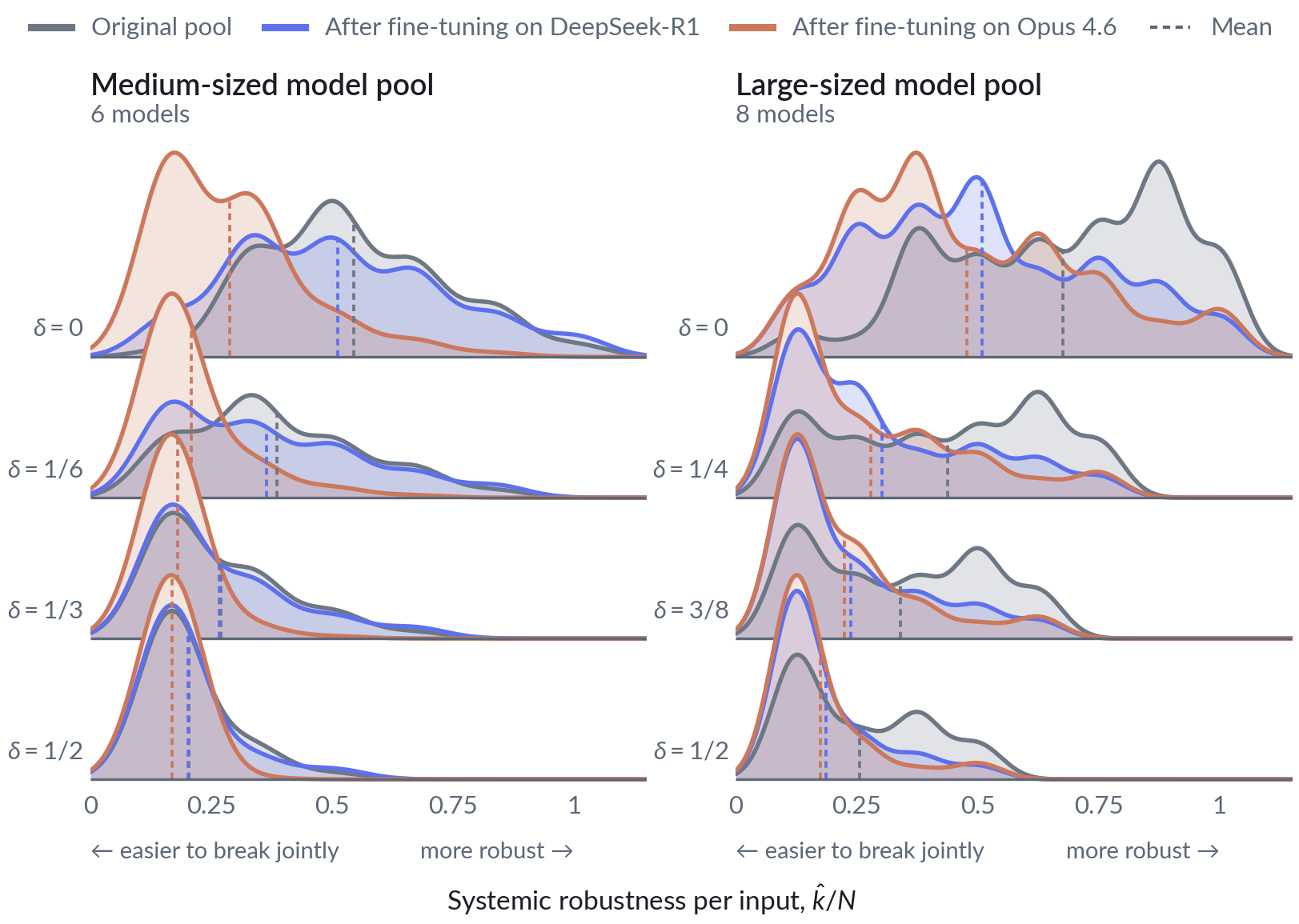}
    \vspace{-10pt}
    \caption{\textbf{Distillation lowers systemic robustness under AmpleGCG + BEAST.} Left: the distribution of singleton-completed $\hat{k}/M$ at varying risk tolerances $\delta$ for the medium-sized pool ($M=6$), comparing the original pool, the pool after DeepSeek-R1 distillation, and the pool after subsequent Opus 4.6 distillation.
    Right: same as left but for the large-sized pool ($M=8$). Lower values mean
    that fewer suffixes cover the pool. Dashed lines denote means.}
    \label{fig:deepseek_opus_distill}
\end{figure*}

\looseness=-1
\xhdr{DeepSeek-R1}
\Cref{fig:deepseek_opus_distill} shows, across the board, that an initial distillation on DeepSeek-R1 traces reduces the systemic robustness under AmpleGCG + BEAST. However, the medium and large pools show slightly different patterns, with the medium pool having a maximum decline in SAR of around $6\%$. The large pool, on the other hand, is impacted more severely, with a decrease in robustness of up to $31\%$. To compromise every model in the medium pool the mean number of attack vectors required moves from $3.258$ to $3.072$, while the mean for the large pool falls from $5.408$ to $4.072$ (\Cref{tab:khat-means}).

In addition, we observe SAR to be a more meaningful indicator of pool robustness than ASR in LLMs, similar to motivating findings on MNIST. \Cref{fig:llm_distillation} shows that SAR erodes faster and further than individual robustness ($1-\mathrm{ASR}$), with $\mathrm{SAR}_{0.85}$ declining by $15.0\%$ compared to a $5.3\%$ decline in individual robustness. Moreover, SAR drops significantly after only $0.05$ epochs, while individual robustness only shows a significant drop after $0.15$ epochs. The final distilled pool is less systemically robust at every level of $\delta$ as can be seen on the right hand side of \Cref{fig:llm_distillation}.

\begin{figure}[h]
    \centering
    \includegraphics[width=\linewidth]{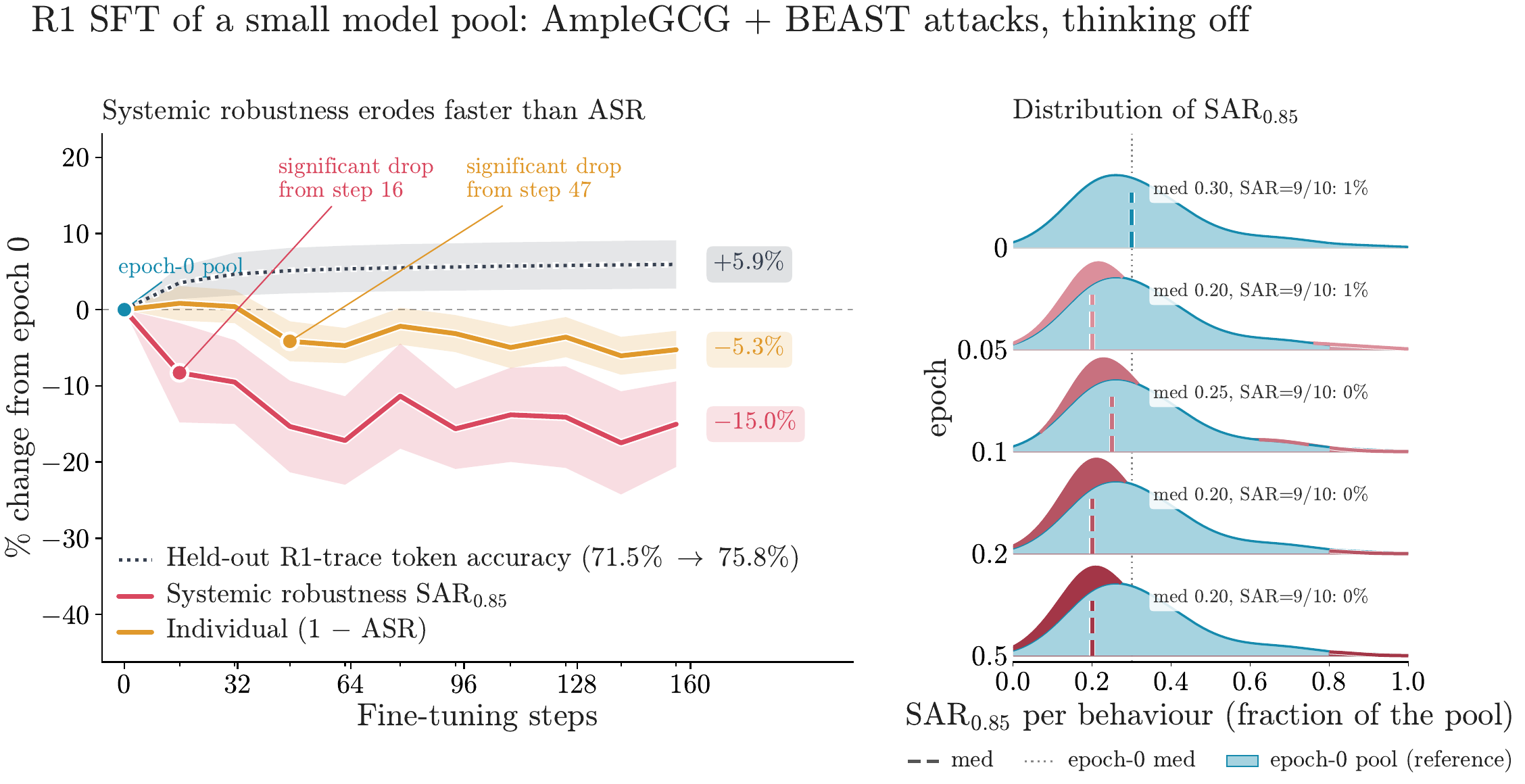}
    \vspace{-10pt}
    \caption{\small \looseness=-1 Same as~\Cref{fig:emnist_distillation} but for distillation of a pool of ten small LLMs on DeepSeek-R1 traces, under AmpleGCG + BEAST with thinking off.}
    \label{fig:llm_distillation}
\end{figure}

\looseness=-1
\xhdr{Opus}
This pattern is amplified when further distilling the descendants of DeepSeek-R1 distillation on Opus 4.6 traces. In \Cref{fig:deepseek_opus_distill} and \Cref{tab:khat-means} it is evident that the mean $\hat{k}$ continues to decline. This time the severity of the pattern inverts: the medium pool loses a further $\sim44\%$ of its systemic robustness, while the large pool becomes only an additional $\sim6\%$ less robust. For $\delta=0$ (i.e., compromising all models) the average attack vectors required fall from $3.072$ to $1.728$ and $4.072$ to $3.824$ for the medium and large pools respectively.

\looseness=-1
\xhdr{Takeaways} 
Across all three pools, fine-tuning on shared teacher traces lowers systemic robustness under AmpleGCG + BEAST (\S\ref{app:llm_pools} for a more detailed analysis). SAR erodes faster and further over the course of distillation than individual robustness, making it a more sensitive indicator of pool-level risk. Neither teacher is consistently more harmful: DeepSeek-R1 distillation degrades the large pool most, whereas subsequent Opus 4.6 distillation degrades the medium pool most.

\section{Related Works}
\label{sec:related_work}

\looseness=-1
\xhdr{Adversarial Attacks}
We can broadly divide adversarial attacks into white-box and black-box attacks. Prominent white-box attacks include the Fast Gradient Sign Method (FGSM) \citep{goodfellow2015explainingharnessingadversarialexamples}, its iterative counterpart Projected Gradient Descent (PGD) \citep{madry2019deeplearningmodelsresistant}, MI-FGSM \citep{dong2018boosting}, the Carlini-Wagner attack \citep{carlini2017towards}, elastic-net attacks \citep{chen2018ead}, and fast adaptive boundary attacks \citep{croce2020minimally}. Black-box attacks commonly rely on the transferability of white-box attacks, e.g., Square Attack \citep{andriushchenko2020squareattackqueryefficientblackbox}, and HopSkipJumpAttack \citep{chen2020hopskipjumpattack}. AutoAttack \citep{croce2020reliable} applies white- and black-box attacks alike to automatically find suitable adversarial attacks.

\looseness=-1
\xhdr{Ensembles in Adversarial Contexts}
Ensembles for adversarial defense, and analogously for attacks, are extensively studied~\citep{li2025quantifyingcorrelationsmachinelearning, wood2023unified}. ADP \citep{pang2019improving} regularises training based on the Shannon entropy, TRS \citep{yang2021trs} enforces low loss-gradient similarity to reduce transferability among members, and \cite{cui2022synergyofexperts} trains ensemble members to minimise vulnerability overlap. GAL \citep{kariyappa2019improvingadversarialrobustnessensembles} reduces the dimensionality of the subspace adversarial samples live in, while DVERGE \citep{yang2020dverge} uses a novel diversity metric to reduce shared vulnerability.~\cite{li2025quantifyingcorrelationsmachinelearning} provides a theoretical grounding and proposes iGAT.

\looseness=-1
\xhdr{Transferability of Adversarial Attacks.}
Transferability was observed in seminal work by \cite{szegedy2014intriguingpropertiesneuralnetworks}, subsequently leveraged in \cite{papernot2016transferabilitymachinelearningphenomena}, and expanded to substantially larger models by \cite{liu2017delvingtransferableadversarialexamples} and first theoretically analyzed in~\cite{demontis2019adversarial}. Some works show that flat adversarial samples yield more transferable attacks~\citep{fan2025bad, zheng2026seekingflatminimaoverdiversesurrogates, ge2023boostingadversarialtransferabilityachieving}. For foundation models, \cite{wei2022towards} introduces transferable attacks on vision transformers, while \cite{zou2023universaltransferableadversarialattacks} and \cite{liao2024amplegcglearninguniversaltransferable} show that attacks transfer across aligned LLMs, with \cite{ma2024improvingadversarialtransferabilitymodel} further showing that alignment itself can increase transferability.

\section{Conclusion}
\label{sec:conclusion}
\looseness=-1
In this work, we introduce Systemic Adversarial Robustness (\nameshort), which measures how many distinct attacks are needed to compromise a target fraction of a model pool. We additionally provide a simple algorithm to compute the \nameshort and analyze it by analogy to the set-cover problem. Empirically, we demonstrate that models (from MLPs on MNIST to modern foundation models) often do share a small number of exploitable vulnerabilities. Further, we show that distillation from a shared model exacerbates the shared vulnerabilities and thus the \nameshort. Finally, while we provide a principled measure, characterizing shared vulnerabilities remains an important and understudied problem.

\subsubsection*{Acknowledgments}

LE, MW and LS acknowledge the support of the UKRI Prosperity Partnership Scheme (FAIR) under EPSRC
Grant EP/V056883/1.

\subsection*{AI use statement}

In this work, we used generative AI tools for no tasks with required disclosure.
We have not used generative AI tools for any tasks with required disclosure.
Additionally, we used generative AI tools for creating or modifying scientific figures or images, creating or editing software code, brainstorming, editing a research paper to improve readability, and identifying relevant literature. We have reviewed all AI-assisted work. We reviewed LLM-generated code to ensure expected behavior and have validated the legitimacy of any identified relevant literature. We take responsibility for the final content of this work,
including text, claims or artifacts produced with the aid of generative AI.

\subsection*{Reproducibility statement}

All of our experiments are seeded and reproducible, with data used to generate figures preserved. However, while we commit to releasing the code after the review period to reproduce our findings upon publication, we do not intend to make the used attack vectors public due to the nature of this work.

\bibliography{iclr2027_conference}

@inproceedings{fan2025bad,
      title={Transferability Bound Theory: Exploring Relationship between Adversarial Transferability and Flatness}, 
      author={Mingyuan Fan and Xiaodan Li and Cen Chen and Wenmeng Zhou and Yaliang Li},
  booktitle    = {Proc. of NeurIPS},
  year         = {2024},
}

@misc{ma2024improvingadversarialtransferabilitymodel,
      title={Improving Adversarial Transferability via Model Alignment}, 
      author={Avery Ma and Amir-massoud Farahmand and Yangchen Pan and Philip Torr and Jindong Gu},
      year={2024},
      eprint={2311.18495},
      archivePrefix={arXiv},
      primaryClass={cs.LG},
      url={https://arxiv.org/abs/2311.18495}, 
}

@ARTICLE{zheng2026seekingflatminimaoverdiversesurrogates,
  author={Zheng, Meixi and Wu, Kehan and Fan, Yanbo and Huang, Rui and Wu, Baoyuan},
  journal={IEEE Transactions on Pattern Analysis and Machine Intelligence}, 
  title={Seeking Flat Minima Over Diverse Surrogates for Improved Adversarial Transferability: A Theoretical Framework and Algorithmic Instantiation}, 
  year={2026},
  volume={},
  number={},
  pages={1-18},
  doi={10.1109/TPAMI.2026.3679726}}

@misc{li2025quantifyingcorrelationsmachinelearning,
      title={Quantifying Correlations of Machine Learning Models}, 
      author={Yuanyuan Li and Neeraj Sarna and Yang Lin},
      year={2025},
      eprint={2502.03937},
      archivePrefix={arXiv},
      primaryClass={cs.LG},
      url={https://arxiv.org/abs/2502.03937}, 
}

@inproceedings{cui2022synergyofexperts,
 author = {Cui, Sen and ZHANG, Jingfeng and Liang, Jian and Han, Bo and Sugiyama, Masashi and Zhang, Changshui},
 booktitle = {Advances in Neural Information Processing Systems},
 editor = {S. Koyejo and S. Mohamed and A. Agarwal and D. Belgrave and K. Cho and A. Oh},
 pages = {32552--32567},
 publisher = {Curran Associates, Inc.},
 title = {Synergy-of-Experts: Collaborate to Improve Adversarial Robustness},
 url = {https://proceedings.neurips.cc/paper_files/paper/2022/file/d16152d53088ad779ffa634e7bf66166-Paper-Conference.pdf},
 volume = {35},
 year = {2022}
}

@misc{zou2023universaltransferableadversarialattacks,
      title={Universal and Transferable Adversarial Attacks on Aligned Language Models}, 
      author={Andy Zou and Zifan Wang and Nicholas Carlini and Milad Nasr and J. Zico Kolter and Matt Fredrikson},
      year={2023},
      eprint={2307.15043},
      archivePrefix={arXiv},
      primaryClass={cs.CL},
      url={https://arxiv.org/abs/2307.15043}, 
}

@misc{liao2024amplegcglearninguniversaltransferable,
      title={AmpleGCG: Learning a Universal and Transferable Generative Model of Adversarial Suffixes for Jailbreaking Both Open and Closed LLMs}, 
      author={Zeyi Liao and Huan Sun},
      year={2024},
      eprint={2404.07921},
      archivePrefix={arXiv},
      primaryClass={cs.CL},
      url={https://arxiv.org/abs/2404.07921}, 
}

@article{szpruch2025insuring,
  title={Insuring AI: Incentivising Safe and Secure Deployment of AI Workflows},
  author={Szpruch, Lukasz and Orfanoudaki, Agni and Maple, Carsten and Wicker, Matthew and Bengio, Yoshua and Lam, Kwok-Yan and Detyniecki, Marcin},
  journal={Available at SSRN 5505759},
  year={2025}
}

@article{bommasani2021opportunities,
  title={On the opportunities and risks of foundation models},
  author={Bommasani, Rishi and Hudson, Drew A and Adeli, Ehsan and Altman, Russ and Arora, Simran and von Arx, Sydney and Bernstein, Michael S and Bohg, Jeannette and Bosselut, Antoine and Brunskill, Emma and others},
  journal={arXiv preprint arXiv:2108.07258},
  year={2021}
}

@ARTICLE{yutysever2020surveryofautonomousdriving,
  author={Yurtsever, Ekim and Lambert, Jacob and Carballo, Alexander and Takeda, Kazuya},
  journal={IEEE Access}, 
  title={A Survey of Autonomous Driving: Common Practices and Emerging Technologies}, 
  year={2020},
  volume={8},
  number={},
  pages={58443-58469},
  doi={10.1109/ACCESS.2020.2983149}}

@misc{goodfellow2015explainingharnessingadversarialexamples,
      title={Explaining and Harnessing Adversarial Examples}, 
      author={Ian J. Goodfellow and Jonathon Shlens and Christian Szegedy},
      year={2015},
      eprint={1412.6572},
      archivePrefix={arXiv},
      primaryClass={stat.ML},
      url={https://arxiv.org/abs/1412.6572}, 
}

@misc{madry2019deeplearningmodelsresistant,
      title={Towards Deep Learning Models Resistant to Adversarial Attacks}, 
      author={Aleksander Madry and Aleksandar Makelov and Ludwig Schmidt and Dimitris Tsipras and Adrian Vladu},
      year={2019},
      eprint={1706.06083},
      archivePrefix={arXiv},
      primaryClass={stat.ML},
      url={https://arxiv.org/abs/1706.06083}, 
}

@article{kumar2024amplegcg,
  title={Amplegcg-plus: A strong generative model of adversarial suffixes to jailbreak llms with higher success rates in fewer attempts},
  author={Kumar, Vishal and Liao, Zeyi and Jones, Jaylen and Sun, Huan},
  journal={arXiv preprint arXiv:2410.22143},
  year={2024}
}

@article{mazeika2024harmbench,
  title={Harmbench: A standardized evaluation framework for automated red teaming and robust refusal},
  author={Mazeika, Mantas and Phan, Long and Yin, Xuwang and Zou, Andy and Wang, Zifan and Mu, Norman and Sakhaee, Elham and Li, Nathaniel and Basart, Steven and Li, Bo and others},
  journal={arXiv preprint arXiv:2402.04249},
  year={2024}
}

@inproceedings{katz2017reluplex,
  title={Reluplex: An efficient SMT solver for verifying deep neural networks},
  author={Katz, Guy and Barrett, Clark and Dill, David L and Julian, Kyle and Kochenderfer, Mykel J},
  booktitle={International conference on computer aided verification},
  pages={97--117},
  year={2017},
  organization={Springer}
}

@inproceedings{carlini2017towards,
  title={Towards evaluating the robustness of neural networks},
  author={Carlini, Nicholas and Wagner, David},
  booktitle={2017 ieee symposium on security and privacy (sp)},
  pages={39--57},
  year={2017},
  organization={Ieee}
}

@inproceedings{chen2018ead,
  title={Ead: elastic-net attacks to deep neural networks via adversarial examples},
  author={Chen, Pin-Yu and Sharma, Yash and Zhang, Huan and Yi, Jinfeng and Hsieh, Cho-Jui},
  booktitle={Proceedings of the AAAI conference on artificial intelligence},
  volume={32},
  number={1},
  year={2018}
}

@inproceedings{croce2020minimally,
  title={Minimally distorted adversarial examples with a fast adaptive boundary attack},
  author={Croce, Francesco and Hein, Matthias},
  booktitle={International conference on machine learning},
  pages={2196--2205},
  year={2020},
  organization={PMLR}
}

@misc{andriushchenko2020squareattackqueryefficientblackbox,
      title={Square Attack: a query-efficient black-box adversarial attack via random search}, 
      author={Maksym Andriushchenko and Francesco Croce and Nicolas Flammarion and Matthias Hein},
      year={2020},
      eprint={1912.00049},
      archivePrefix={arXiv},
      primaryClass={cs.LG},
      url={https://arxiv.org/abs/1912.00049}, 
}

@inproceedings{chen2020hopskipjumpattack,
  title={Hopskipjumpattack: A query-efficient decision-based attack},
  author={Chen, Jianbo and Jordan, Michael I and Wainwright, Martin J},
  booktitle={2020 ieee symposium on security and privacy (sp)},
  pages={1277--1294},
  year={2020},
  organization={IEEE}
}

@article{wood2023unified,
  title={A unified theory of diversity in ensemble learning},
  author={Wood, Danny and Mu, Tingting and Webb, Andrew M and Reeve, Henry WJ and Lujan, Mikel and Brown, Gavin},
  journal={Journal of machine learning research},
  volume={24},
  number={359},
  pages={1--49},
  year={2023}
}

@inproceedings{pang2019improving,
  title={Improving adversarial robustness via promoting ensemble diversity},
  author={Pang, Tianyu and Xu, Kun and Du, Chao and Chen, Ning and Zhu, Jun},
  booktitle={International Conference on Machine Learning},
  pages={4970--4979},
  year={2019},
  organization={PMLR}
}

@article{yang2021trs,
  title={Trs: Transferability reduced ensemble via promoting gradient diversity and model smoothness},
  author={Yang, Zhuolin and Li, Linyi and Xu, Xiaojun and Zuo, Shiliang and Chen, Qian and Zhou, Pan and Rubinstein, Benjamin and Zhang, Ce and Li, Bo},
  journal={Advances in Neural Information Processing Systems},
  volume={34},
  pages={17642--17655},
  year={2021}
}

@misc{kariyappa2019improvingadversarialrobustnessensembles,
      title={Improving Adversarial Robustness of Ensembles with Diversity Training}, 
      author={Sanjay Kariyappa and Moinuddin K. Qureshi},
      year={2019},
      eprint={1901.09981},
      archivePrefix={arXiv},
      primaryClass={stat.ML},
      url={https://arxiv.org/abs/1901.09981}, 
}

@article{yang2020dverge,
  title={Dverge: diversifying vulnerabilities for enhanced robust generation of ensembles},
  author={Yang, Huanrui and Zhang, Jingyang and Dong, Hongliang and Inkawhich, Nathan and Gardner, Andrew and Touchet, Andrew and Wilkes, Wesley and Berry, Heath and Li, Hai},
  journal={Advances in Neural Information Processing Systems},
  volume={33},
  pages={5505--5515},
  year={2020}
}

@misc{szegedy2014intriguingpropertiesneuralnetworks,
      title={Intriguing properties of neural networks}, 
      author={Christian Szegedy and Wojciech Zaremba and Ilya Sutskever and Joan Bruna and Dumitru Erhan and Ian Goodfellow and Rob Fergus},
      year={2014},
      eprint={1312.6199},
      archivePrefix={arXiv},
      primaryClass={cs.CV},
      url={https://arxiv.org/abs/1312.6199}, 
}

@misc{papernot2016transferabilitymachinelearningphenomena,
      title={Transferability in Machine Learning: from Phenomena to Black-Box Attacks using Adversarial Samples}, 
      author={Nicolas Papernot and Patrick McDaniel and Ian Goodfellow},
      year={2016},
      eprint={1605.07277},
      archivePrefix={arXiv},
      primaryClass={cs.CR},
      url={https://arxiv.org/abs/1605.07277}, 
}

@misc{liu2017delvingtransferableadversarialexamples,
      title={Delving into Transferable Adversarial Examples and Black-box Attacks}, 
      author={Yanpei Liu and Xinyun Chen and Chang Liu and Dawn Song},
      year={2017},
      eprint={1611.02770},
      archivePrefix={arXiv},
      primaryClass={cs.LG},
      url={https://arxiv.org/abs/1611.02770}, 
}

@inproceedings{demontis2019adversarial,
  title={Why do adversarial attacks transfer? explaining transferability of evasion and poisoning attacks},
  author={Demontis, Ambra and Melis, Marco and Pintor, Maura and Jagielski, Matthew and Biggio, Battista and Oprea, Alina and Nita-Rotaru, Cristina and Roli, Fabio},
  booktitle={28th USENIX security symposium (USENIX security 19)},
  pages={321--338},
  year={2019}
}

@inproceedings{wei2022towards,
  title={Towards transferable adversarial attacks on vision transformers},
  author={Wei, Zhipeng and Chen, Jingjing and Goldblum, Micah and Wu, Zuxuan and Goldstein, Tom and Jiang, Yu-Gang},
  booktitle={Proceedings of the AAAI conference on artificial intelligence},
  volume={36},
  number={3},
  pages={2668--2676},
  year={2022}
}

@misc{ge2023boostingadversarialtransferabilityachieving,
      title={Boosting Adversarial Transferability by Achieving Flat Local Maxima}, 
      author={Zhijin Ge and Hongying Liu and Xiaosen Wang and Fanhua Shang and Yuanyuan Liu},
      year={2023},
      eprint={2306.05225},
      archivePrefix={arXiv},
      primaryClass={cs.CV},
      url={https://arxiv.org/abs/2306.05225}, 
}

@inproceedings{croce2020reliable,
  title={Reliable evaluation of adversarial robustness with an ensemble of diverse parameter-free attacks},
  author={Croce, Francesco and Hein, Matthias},
  booktitle={International conference on machine learning},
  pages={2206--2216},
  year={2020},
  organization={PMLR}
}

@article{
algorithmicmonoculture,
author = {Jon Kleinberg  and Manish Raghavan },
title = {Algorithmic monoculture and social welfare},
journal = {Proceedings of the National Academy of Sciences},
volume = {118},
number = {22},
pages = {e2018340118},
year = {2021},
doi = {10.1073/pnas.2018340118},
URL = {https://www.pnas.org/doi/abs/10.1073/pnas.2018340118},
eprint = {https://www.pnas.org/doi/pdf/10.1073/pnas.2018340118}}

@inproceedings{szegedy2014intriguing,
  title     = {Intriguing Properties of Neural Networks},
  author    = {Szegedy, Christian and Zaremba, Wojciech and Sutskever, Ilya and Bruna, Joan and Erhan, Dumitru and Goodfellow, Ian and Fergus, Rob},
  booktitle = {International Conference on Learning Representations (ICLR)},
  year      = {2014},
  url       = {https://arxiv.org/abs/1312.6199}
}

@article{papernot2016transferability,
  title   = {Transferability in Machine Learning: From Phenomena to Black-Box Attacks Using Adversarial Samples},
  author  = {Papernot, Nicolas and McDaniel, Patrick and Goodfellow, Ian},
  journal = {arXiv preprint arXiv:1605.07277},
  year    = {2016}
}

@inproceedings{liu2017delving,
  title     = {Delving into Transferable Adversarial Examples and Black-box Attacks},
  author    = {Liu, Yanpei and Chen, Xinyun and Liu, Chang and Song, Dawn},
  booktitle = {International Conference on Learning Representations (ICLR)},
  year      = {2017},
  url       = {https://arxiv.org/abs/1611.02770}
}

@inproceedings{dong2018boosting,
  title     = {Boosting Adversarial Attacks with Momentum},
  author    = {Dong, Yinpeng and Liao, Fangzhou and Pang, Tianyu and Su, Hang and Zhu, Jun and Hu, Xiaolin and Li, Jianguo},
  booktitle = {Proceedings of the IEEE Conference on Computer Vision and Pattern Recognition (CVPR)},
  pages     = {9185--9193},
  year      = {2018}
}

@inproceedings{mao2022transfer,
  title     = {Transfer Attacks Revisited: A Large-Scale Empirical Study in Real Computer Vision Settings},
  author    = {Mao, Yuhao and Fu, Chong and Wang, Saizhuo and Ji, Shouling and Zhang, Xuhong and Liu, Zhenguang and Zhou, Jun and Liu, Alex X. and Beyah, Raheem and Wang, Ting},
  booktitle = {2022 IEEE Symposium on Security and Privacy (SP)},
  pages     = {1423--1439},
  publisher = {IEEE},
  year      = {2022},
  doi       = {10.1109/SP46214.2022.9833783}
}

@inproceedings{chen2024rethinking,
  title     = {Rethinking Model Ensemble in Transfer-based Adversarial Attacks},
  author    = {Chen, Huanran and Zhang, Yichi and Dong, Yinpeng and Yang, Xiao and Su, Hang and Zhu, Jun},
  booktitle = {International Conference on Learning Representations (ICLR)},
  year      = {2024},
  url       = {https://arxiv.org/abs/2303.09105}
}

@article{zou2023universal,
  title   = {Universal and Transferable Adversarial Attacks on Aligned Language Models},
  author  = {Zou, Andy and Wang, Zifan and Carlini, Nicholas and Nasr, Milad and Kolter, J. Zico and Fredrikson, Matt},
  journal = {arXiv preprint arXiv:2307.15043},
  year    = {2023}
}

@article{slavik1997improved,
  title={Improved performance of the greedy algorithm for partial cover},
  author={Slav{\'\i}k, Petr},
  journal={Information Processing Letters},
  volume={64},
  number={5},
  pages={251--254},
  year={1997},
  publisher={Elsevier}
}

@book{kearns1994computational,
author = {Kearns, Michael J. and Vazirani, Umesh V.},
title = {An introduction to computational learning theory},
year = {1994},
isbn = {0262111934},
publisher = {MIT Press},
address = {Cambridge, MA, USA}
}

@article{chvatal1979greedy,
  title={A greedy heuristic for the set-covering problem},
  author={Chvatal, Vasek},
  journal={Mathematics of operations research},
  volume={4},
  number={3},
  pages={233--235},
  year={1979},
  publisher={INFORMS}
}

@article{goundan2007revisiting,
  title={Revisiting the greedy approach to submodular set function maximization},
  author={Goundan, Pranava R and Schulz, Andreas S},
  journal={Optimization online},
  volume={1},
  year={2007}
}

@article{xu2024survey,
  title={A survey on knowledge distillation of large language models},
  author={Xu, Xiaohan and Li, Ming and Tao, Chongyang and Shen, Tao and Cheng, Reynold and Li, Jinyang and Xu, Can and Tao, Dacheng and Zhou, Tianyi},
  journal={arXiv preprint arXiv:2402.13116},
  year={2024}
}

@article{geer2003cyberinsecurity,
  title={CyberInsecurity: The cost of monopoly--how the dominance of Microsoft’s products poses a risk to security},
  author={Geer, Daniel and Bace, Rebecca and Gutmann, Peter and Metzger, Perry and Pfleeger, C and Querterman, J and Scheier, B},
  journal={Computer \& Communications Industry Association Report},
  year={2003}
}

@article{toups2023ecosystem,
  title={Ecosystem-level analysis of deployed machine learning reveals homogeneous outcomes},
  author={Toups, Connor and Bommasani, Rishi and Creel, Kathleen and Bana, Sarah and Jurafsky, Dan and Liang, Percy S},
  journal={Advances in Neural Information Processing Systems},
  volume={36},
  pages={51178--51201},
  year={2023}
}

@article{goel2025great,
  title={Great models think alike and this undermines ai oversight},
  author={Goel, Shashwat and Struber, Joschka and Auzina, Ilze Amanda and Chandra, Karuna K and Kumaraguru, Ponnurangam and Kiela, Douwe and Prabhu, Ameya and Bethge, Matthias and Geiping, Jonas},
  journal={arXiv preprint arXiv:2502.04313},
  year={2025}
}

@article{kim2025correlated,
  title={Correlated errors in large language models},
  author={Kim, Elliot and Garg, Avi and Peng, Kenny and Garg, Nikhil},
  journal={arXiv preprint arXiv:2506.07962},
  year={2025}
}

@article{guo2025deepseek,
  title={DeepSeek-R1 incentivizes reasoning in LLMs through reinforcement learning},
  author={Guo, Daya and Yang, Dejian and Zhang, Haowei and Song, Junxiao and Wang, Peiyi and Zhu, Qihao and Xu, Runxin and Zhang, Ruoyu and Ma, Shirong and Bi, Xiao and others},
  journal={Nature},
  volume={645},
  number={8081},
  pages={633--638},
  year={2025},
  publisher={Nature Publishing Group UK London}
}

@article{cloud2025subliminal,
  title={Subliminal learning: Language models transmit behavioral traits via hidden signals in data},
  author={Cloud, Alex and Le, Minh and Chua, James and Betley, Jan and Sztyber-Betley, Anna and Hilton, Jacob and Marks, Samuel and Evans, Owain},
  journal={arXiv preprint arXiv:2507.14805},
  year={2025}
}

@article{bai2023qwen,
  title={Qwen technical report},
  author={Bai, Jinze and Bai, Shuai and Chu, Yunfei and Cui, Zeyu and Dang, Kai and Deng, Xiaodong and Fan, Yang and Ge, Wenbin and Han, Yu and Huang, Fei and others},
  journal={arXiv preprint arXiv:2309.16609},
  year={2023}
}

@article{grattafiori2024llama,
  title={The llama 3 herd of models},
  author={Grattafiori, Aaron and Dubey, Abhimanyu and Jauhri, Abhinav and Pandey, Abhinav and Kadian, Abhishek and Al-Dahle, Ahmad and Letman, Aiesha and Mathur, Akhil and Schelten, Alan and Vaughan, Alex and others},
  journal={arXiv preprint arXiv:2407.21783},
  year={2024}
}

@article{sadasivan2024fast,
  title={Fast adversarial attacks on language models in one gpu minute},
  author={Sadasivan, Vinu Sankar and Saha, Shoumik and Sriramanan, Gaurang and Kattakinda, Priyatham and Chegini, Atoosa and Feizi, Soheil},
  journal={arXiv preprint arXiv:2402.15570},
  year={2024}
}

@inproceedings{liu2025autodan,
  title={Autodan-turbo: A lifelong agent for strategy self-exploration to jailbreak llms},
  author={Liu, Xiaogeng and Li, Peiran and Suh, G Edward and Vorobeychik, Yevgeniy and Mao, Zhuoqing and Jha, Somesh and McDaniel, Patrick and Sun, Huan and Li, Bo and Xiao, Chaowei},
  booktitle={International Conference on Learning Representations},
  volume={2025},
  pages={10313--10360},
  year={2025}
}

@inproceedings{sun2025iterative,
  title={Iterative self-tuning llms for enhanced jailbreaking capabilities},
  author={Sun, Chung-En and Liu, Xiaodong and Yang, Weiwei and Weng, Tsui-Wei and Cheng, Hao and San, Aidan and Galley, Michel and Gao, Jianfeng},
  booktitle={Proceedings of the 2025 Conference of the Nations of the Americas Chapter of the Association for Computational Linguistics: Human Language Technologies (Volume 1: Long Papers)},
  pages={5768--5786},
  year={2025}
}
\bibliographystyle{abbrvnat}

\clearpage
\appendix
\section{Appendix}

\subsection{Proofs}

\approxratio*
\begin{proof}
The lower bound holds because $X$ satisfies the constraint in \eqref{eq:psc}. For the upper bound let $X^\star$ attain $k^\star$ and set $U^\star := \bigcup_{x' \in X^\star} \mathcal C(x')$, so $|U^\star| \geq \tau$. Write $\mathcal B_t$ for the covered set after $t$ iterations and $r_t := \max\{0, \tau - |\mathcal B_t|\}$ for the remaining requirement, with $r_0 = \tau$. As noted above the fallback is never executed, so at iteration $t$ the uncovered set is $\mathcal I_{t-1} = [M] \setminus \mathcal B_{t-1}$ and contains at least $|U^\star| - |\mathcal B_{t-1}| \geq r_{t-1}$ elements of $U^\star$. These are covered by the $k^\star$ vectors of $X^\star$, so some $x^\star \in X^\star$ satisfies $|\mathcal C(x^\star) \cap \mathcal I_{t-1}| \geq r_{t-1} / k^\star$, and by Definition~\ref{def:alpha} the accepted vector covers $|F_t| \geq \alpha\, r_{t-1} / k^\star$ new models. Hence $r_t \leq r_{t-1}(1 - \alpha / k^\star)$ and, iterating, $r_t \leq \tau (1 - \alpha / k^\star)^t < \tau e^{-\alpha t / k^\star}$ for all $t \geq 1$. The right-hand side is at most $1$ once $t \geq (k^\star / \alpha) \ln \tau$, at which point $r_t < 1$ and, being an integer, $r_t = 0$. The loop therefore stops at $\hat k \leq \max\{1, \lceil (k^\star / \alpha) \ln \tau \rceil\} \leq k^\star + (k^\star / \alpha) \ln \tau$, using $k^\star \geq 1$.
\end{proof}

\subsection{Systemic Risk Evaluation Algorithms}

\xhdr{Joint PGD for Image Classification Models}

For an image $x \in [0,1]^d$, with pixel intensities scaled to $[0,1]$, the
admissible perturbation set is the $\ell_\infty$ ball
$\mathcal{B}_\varepsilon(x)$ within the valid pixel range,
\begin{equation}
  \mathcal{B}_\varepsilon(x)
      = \bigl\{\, x' \in [0,1]^d : \|x' - x\|_\infty \le \varepsilon \,\bigr\}.
\end{equation}

\paragraph{Objective.}
Let $z_j(x') \in \mathbb{R}^{|\mathcal{Y}|}$ denote the
logits of $f_j$, so that $f_j(x') = \arg\max_{c} z_j(x')_c$. Given an instance
$x$ and a set of uncovered models $\mathcal{I} \subseteq [M]$, joint
PGD seeks a single admissible input that compromises every uncovered model at
once, by maximizing their mean cross-entropy with respect to their own clean
predictions:
\begin{equation}
  \label{eq:joint-loss}
  \max_{x' \in T_x} \;
  \mathcal{L}_{\mathcal{I}}(x'),
  \qquad
  \mathcal{L}_{\mathcal{I}}(x') = \frac{1}{|\mathcal{I}|}
  \sum_{j \in \mathcal{I}} \mathrm{CE}\bigl(z_j(x'),\, f_j(x)\bigr).
\end{equation}
This is an untargeted, white-box attack: every model in $\mathcal{I}$
contributes an equal share of the gradient, and models outside
$\mathcal{I}$, which are already covered, contribute nothing.

\paragraph{Optimization.}
We solve \eqref{eq:joint-loss} with $\ell_\infty$ projected gradient ascent.
Let $\Pi_{\mathcal B_\varepsilon(x)}$ be the projection onto $\mathcal B_\varepsilon(x)$, which clips each coordinate
$x'_i$ to $[-\varepsilon, \varepsilon]$ around $x_i$ and then to $[0,1]$. Both
constraints are coordinate-wise intervals that intersect, so these two clips
give the exact projection. Starting from a uniform random point
$x'^{(0)} = \Pi_{\mathcal B_\varepsilon(x)}(x + u)$ with
$u \sim \mathcal{U}([-\varepsilon,\varepsilon]^d)$, each step is
\begin{equation}
  \label{eq:pgd-step}
  x'^{(t+1)} = \Pi_{\mathcal B_\varepsilon(x)}\Bigl( x'^{(t)} + \eta \,
  \operatorname{sign}\bigl( \nabla_{x'}\, \mathcal{L}_{\mathcal{I}}(x'^{(t)}) \bigr) \Bigr),
  \qquad t = 0,\dots,n_{\mathrm{it}}-1.
\end{equation}
We run $n_{\mathrm{rs}}$ independent random restarts and return the final
iterate of the restart with the highest objective
$\mathcal{L}_{\mathcal{I}}(x'^{(n_{\mathrm{it}})})$ (\Cref{alg:joint-pgd}).
MNIST attack hyperparameters are shown in \Cref{tab:hp-joint-pgd}.

\begin{algorithm}[t]
\caption{Joint PGD attack $\mathcal{A}_{\mathcal{I}}(x) = \textsc{JointPGD}(x, \mathcal{I})$}
\label{alg:joint-pgd}
\begin{algorithmic}[1]
\Require instance $x \in [0,1]^d$; active models $\mathcal{I} \subseteq [M]$;
         budget $\varepsilon$; steps $n_{\mathrm{it}}$; step size $\eta$;
         restarts $n_{\mathrm{rs}}$
\State $\ell^\star \gets -\infty$
\For{$r = 1, \dots, n_{\mathrm{rs}}$}
  \State $x' \gets \Pi_{T_x}(x + u)$, \quad $u \sim \mathcal{U}([-\varepsilon,\varepsilon]^d)$
  \For{$t = 1, \dots, n_{\mathrm{it}}$}
    \State $x' \gets \Pi_{T_x}\bigl(x' + \eta\,\operatorname{sign}(\nabla_{x'} \mathcal{L}_{\mathcal{I}}(x'))\bigr)$
  \EndFor
  \If{$\mathcal{L}_{\mathcal{I}}(x') > \ell^\star$}
    \State $x^\star \gets x'$; \quad $\ell^\star \gets \mathcal{L}_{\mathcal{I}}(x')$
  \EndIf
\EndFor
\State \Return $x^\star$
\end{algorithmic}
\end{algorithm}

\begin{table}[t]
  \centering
  \caption{Hyperparameters of joint PGD on MNIST.}
  \label{tab:hp-joint-pgd}
  \small
  \begin{tabular}{@{}l p{0.52\linewidth}@{}}
    \toprule
    Hyperparameter & Value \\
    \midrule
    Steps $n_{\mathrm{it}}$ & 40 \\
    Step size $\eta$ & $2.5\,\varepsilon / n_{\mathrm{it}} = \varepsilon/16$ \\
    Random restarts $n_{\mathrm{rs}}$ & 3 \\
    \bottomrule
  \end{tabular}
\end{table}

\xhdr{Joint BEAST for Large Language Models}

For LLMs, an instance $x$ is a harmful request from AdvBench, paired with an
affirmative target response $y_x$ (e.g.\ ``Sure, here is \dots''). The
attacker appends an adversarial suffix $s$ to the request, written
$x \oplus s$ while adhering to an attack budget $L$. A suffix $s$ is conforming with the budget constraint iff $|s|\leq L$ under a fixed canonical tokenizer.

\paragraph{Objective.}
Each model $f_j$ is an autoregressive LLM with its own tokenizer. Let $p_j(\cdot \mid \cdot)$ be its next-token distribution. Following BEAST \citep{sadasivan2024fast}, we optimise a surrogate: the per-token
negative log-likelihood of the target response,
\begin{equation}
  \ell_j(x') = -\frac{1}{|y_x|_j} \sum_{t=1}^{|y_x|_j}
      \log p_j\bigl(y_{x,t} \,\big|\, c_j(x'),\, y_{x,<t}\bigr),
\end{equation}
where $y_x$ is tokenized with $f_j$'s own tokenizer and $|y_x|_j$ is its
length. Given an instance $x$ and a set of uncovered models
$\mathcal{I} \subseteq [M]$, joint BEAST \cite{} seeks a suffix that makes every
uncovered model likely to begin its answer with the target, by minimizing
their mean target negative log likelihood (NLL).

\paragraph{Optimization.}
We optimize for the NLL with a gradient-free beam search over suffix
tokens in which every active model proposes candidates
(\Cref{alg:joint-beast}). The beam starts from $B$ seed suffixes of $L_0$
tokens, generated for the request $x$ by AmpleGCG \citep{liao2024amplegcglearninguniversaltransferable}. The search then extends
every suffix by one token per step for $L - L_0$ steps. At each step, every
pair of beam suffix $s$ and active model $j \in \mathcal{I}$ proposes
$n_{j,s}$ distinct tokens, sampled without replacement from
\begin{equation}
  \label{eq:beast-proposal}
  q_{j,s} = \operatorname{softmax}\bigl( z_j(x \oplus s) / T \bigr),
\end{equation}
where $z_j(x \oplus s)$ are the logits of $f_j$ at the last position of
$c_j(x \oplus s)$, that is, where its response would begin. Each proposal is decoded to text, re-encoded with the
canonical tokenizer, appended to $s$, and truncated to $L_0 + t$ canonical
tokens. The quotas $n_{j,s}$ are split as
evenly as possible over all $B\,|\mathcal{I}|$ source pairs, so that no
single model dominates the candidate set. Every candidate is then scored on
every active model, and the $B$ candidates with the lowest mean NLL
$\mathcal{L}_{\mathcal{I}}$ form the next beam. After the final step, the
attack returns all $B$ suffixes in the beam, not only the best one.

\begin{algorithm}[t]
\caption{Joint BEAST attack $\mathcal{A}_{\mathcal{I}}(x) = \textsc{JointBEAST}(x, \mathcal{I})$}
\label{alg:joint-beast}
\begin{algorithmic}[1]
\Require instance $x$ with target $y_x$; active models $\mathcal{I} \subseteq [M]$;
         seed suffixes $s_1,\dots,s_B$ of $L_0$ tokens; suffix length $L$;
         query budget $Q$; temperature $T$
\State $\mathcal{S} \gets \{s_1, \dots, s_B\}$; score each on every $j \in \mathcal{I}$
\For{$t = 1, \dots, L - L_0$}
  \State $C_t \gets$ number of candidates the budget allots to step $t$
  \State $\mathcal{K} \gets \emptyset$
  \For{each $s \in \mathcal{S}$ and $j \in \mathcal{I}$}
    \State draw $n_{j,s}$ distinct tokens $v \sim q_{j,s}$ \Comment{\eqref{eq:beast-proposal}}
    \State add $s \oplus v$, truncated to $L_0 + t$ tokens, to $\mathcal{K}$
  \EndFor
  \State remove duplicates, keeping $C_t$ candidates in $\mathcal{K}$
  \State $\mathcal{S} \gets$ the $B$ suffixes $s \in \mathcal{K}$ with the lowest $\mathcal{L}_{\mathcal{I}}(x \oplus s)$
\EndFor
\State \Return $\{\, x \oplus s : s \in \mathcal{S} \,\}$
\end{algorithmic}
\end{algorithm}

Hyperparameter details for all jointBEAST experiments can be found in \Cref{tab:hp-joint-beast}.

\begin{table}[t]
  \centering
  \caption{Hyperparameters of joint BEAST.}
  \label{tab:hp-joint-beast}
  \small
  \begin{tabular}{@{}l p{0.52\linewidth}@{}}
    \toprule
    Hyperparameter & Value \\
    \midrule
    Beam width $B$ & 10 \\
    Seed length $L_0$ & 20 tokens \\
    Suffix length $L$ & 30 tokens (as per canonical tokenizer) \\
    Seed generator and canonical tokenizer & AmpleGCG-llama2-sourced-llama2-7b-chat \\
    Temperature $T$ & 1 \\
    Query budget per call $Q$ & 15{,}000 \\
    Maximum calls $n_{\mathrm{calls}}$ & 4 (60{,}000 search queries) \\
    Samples & 25 AdvBench behaviors \\
    Judge & HarmBench-Llama-2-13b-cls \\
    \bottomrule
  \end{tabular}
\end{table}

\xhdr{Joint GCG for Large Language Models}

For joint GCG, an instance $x$ is a harmful request from HarmBench. As for
joint BEAST, the attacker appends an adversarial suffix $s$ to the request,
written $x \oplus s$. Unlike joint BEAST, the suffix length is not
constrained explicitly: it is set by the initial suffix, and each
optimization step substitutes a single token.

\paragraph{Objective.}
Joint GCG optimises an untargeted surrogate. Instead of pulling each model
towards an affirmative target, it pushes each model away from a fixed refusal
$r$ (``I cannot assist with that request.''). For model $f_j$, the surrogate
is the per-token negative log-likelihood of the refusal.
Given an instance $x$ and a set of
uncovered models $\mathcal{I} \subseteq [M]$, joint GCG seeks a suffix that
makes every uncovered model unlikely to refuse, by maximizing their mean
refusal NLL:
\begin{equation}
  \label{eq:joint-gcg}
  \max_{s} \;
  \mathcal{L}_{\mathcal{I}}(x \oplus s),
  \qquad
  \mathcal{L}_{\mathcal{I}}(x') = \frac{1}{|\mathcal{I}|}
  \sum_{j \in \mathcal{I}} \ell_j(x').
\end{equation}

\paragraph{Optimization.}
We optimize \eqref{eq:joint-gcg} with greedy coordinate gradient (GCG)
search \citep{zou2023universal}, adapted to a pool of models that do not share
a tokenizer (\Cref{alg:joint-gcg}). Suffixes are therefore exchanged between
models as text and never as token ids. At each step, one model
$j \in \mathcal{I}$ proposes candidates, and the proposing model rotates over
$\mathcal{I}$ so that every active model proposes for an (almost) equal
number of steps. The proposing model tokenizes the current suffix $s$ with
its own tokenizer as $(v_1, \dots, v_n)$ and computes the gradient of its
surrogate $\tilde{\ell}_j$ with respect to the one-hot encoding
$e_i \in \{0,1\}^{|\mathcal{V}_j|}$ of each suffix token $v_i$. Each of $n_{\mathrm{cand}}$ candidates is formed by
choosing a position $i$ uniformly at random and replacing $v_i$ with a token
drawn uniformly from the $n_{\mathrm{top}}$ tokens with the largest gradient
entries at that position. The candidate with the highest mean refusal NLL
$\mathcal{L}_{\mathcal{I}}$ replaces the current suffix if it improves
on the best objective found so far.

\begin{algorithm}[t]
\caption{Joint GCG attack $\mathcal{A}_{\mathcal{I}}(x) = \textsc{JointGCG}(x, \mathcal{I})$}
\label{alg:joint-gcg}
\begin{algorithmic}[1]
\Require instance $x$; active models $\mathcal{I} \subseteq [M]$; initial suffix $s_0$;
         steps $n_{\mathrm{it}}$; candidates per step $n_{\mathrm{cand}}$;
         top-$n_{\mathrm{top}}$ substitutions
\State $s \gets s_0$; \quad $\ell^\star \gets -\infty$
\For{$t = 1, \dots, n_{\mathrm{it}}$}
  \State $j \gets$ proposing model for step $t$ \Comment{rotates over $\mathcal{I}$}
  \State tokenize $s$ with $f_j$'s tokenizer as $(v_1, \dots, v_n)$
  \State $g_i \gets \nabla_{e_i}\, \tilde{\ell}_j(x \oplus s)$ for $i = 1, \dots, n$
  \State $\mathcal{K} \gets \{s\}$
  \For{$c = 1, \dots, n_{\mathrm{cand}}$}
    \State $i \sim \mathcal{U}(\{1, \dots, n\})$; \quad $v \sim \mathcal{U}\bigl(\operatorname{top}_{n_{\mathrm{top}}}(g_i)\bigr)$
    \State add the text of $(v_1, \dots, v_{i-1}, v, v_{i+1}, \dots, v_n)$ to $\mathcal{K}$
  \EndFor
  \State $s' \gets \arg\max_{u \in \mathcal{K}} \tilde{\mathcal{L}}_{\mathcal{I}}(x \oplus u)$
  \If{$\tilde{\mathcal{L}}_{\mathcal{I}}(x \oplus s') > \ell^\star$}
    \State $s \gets s'$; \quad $\ell^\star \gets \tilde{\mathcal{L}}_{\mathcal{I}}(x \oplus s')$
  \EndIf
\EndFor
\State \Return $x \oplus s$
\end{algorithmic}
\end{algorithm}

In all joint GCG experiments we use the hyperparameter settings shown in \Cref{tab:hp-joint-gcg}.

\begin{table}[t]
  \centering
  \caption{Hyperparameters of joint GCG.}
  \label{tab:hp-joint-gcg}
  \small
  \begin{tabular}{@{}l p{0.55\linewidth}@{}}
    \toprule
    Hyperparameter & Value \\
    \midrule
    Steps per call $n_{\mathrm{it}}$ & 20 \\
    Candidates per step $n_{\mathrm{cand}}$ & 64 \\
    Top-$n_{\mathrm{top}}$ substitutions & 256 \\
    Samples & 25 HarmBench behaviors \\
    Judge & HarmBench-Llama-2-13b-cls \\
    \bottomrule
  \end{tabular}
\end{table}

\xhdr{Joint AutoDAN for Large Language Models}

For joint AutoDAN, an instance $x$ is a harmful request from HarmBench.
Unlike joint BEAST and joint GCG, the attack does not append a suffix.
Instead, an attacker LLM $g$ rewrites the request into a complete
natural-language prompt $u$ that preserves its objective. A prompt is
admissible if it is at most 6{,}000 characters long, reads as ordinary
natural language (at least five words and at least 85\% ASCII characters),
and does not duplicate an earlier prompt up to whitespace and case.

\paragraph{Objective.}
The attack treats the target models as black boxes and uses only their text
responses. For each prompt $u$, every model in the pool generates a greedy
response $f_j(u)$, which the HarmBench classifier labels with a binary
outcome $b_j(u) \in \{0,1\}$ and a confidence
$\sigma_j(u) \in [0,1]$. The confidence is the classifier's normalised probability of
answering "yes" or "no" at its first output token. Given an instance $x$ and a set of models
$\mathcal{I} \subseteq [M]$, joint AutoDAN seeks a prompt that compromises as
many models in $\mathcal{I}$ as possible. It ranks prompts lexicographically by
\begin{equation}
  \label{eq:joint-autodan}
  R_{\mathcal{I}}(u) = \Bigl(\,
    \sum_{j \in \mathcal{I}} b_j(u),\;\;
    \frac{1}{|\mathcal{I}|} \sum_{j \in \mathcal{I}} \sigma_j(u),\;\;
    q_{0.2}\bigl(\{\sigma_j(u)\}_{j \in \mathcal{I}}\bigr)
  \Bigr),
\end{equation}
that is, by the number of compromised models, then the mean confidence, then
the lower 20\% quantile of the confidences. The mean and the quantile break
ties between prompts that compromise the same number of models, and favour
prompts that move the least-compromised models.

\paragraph{Optimization.}
We optimize \eqref{eq:joint-autodan} with a strategy-guided search adapted
from AutoDAN-Turbo \citep{liu2025autodan}, in which the attacker revises
prompts using feedback from one \emph{anchor} model at a  in a round robin manner
(\Cref{alg:joint-autodan}). The attacker is guided by a frozen library of
abstract jailbreak strategies $\Sigma$ from the AutoDAN-Turbo warm-up library. The search first evaluates the unmodified request $x$ on every model.
The attacker then produces a cold-start rewrite of $x$ without strategy
guidance, which becomes the first candidate. The anchor then cycles once over
all $M$ models in the pool. For anchor $a$, two parents are selected from all
candidates found so far: the \emph{systemic} parent $u^{\mathrm{sys}}$, which
maximizes $R_{\mathcal{I}}$, and the \emph{specialist} parent
$u^{\mathrm{spec}}$, which is best on the anchor alone, ranked by
$(b_a, \sigma_a, R_{\mathcal{I}})$. The strategies in $\Sigma$ are ranked by
lexical cosine similarity between their definitions and the request, the
systemic parent, and the anchor's response to it. The attacker then generates
$n_{\mathrm{p}}$ proposals. Proposals alternate between the systemic and the
specialist parent. Each of the first $n_{\mathrm{p}} - 1$ proposals applies one
of the top-ranked strategies, and the last combines the top two. For every
proposal, the attacker receives the request, the parent prompt, the anchor's
response to the systemic parent, whether that response was judged harmful,
and the strategy definitions. It returns a revised prompt, which is resampled
up to three times if it is not admissible. Every proposal is evaluated on
every model in the pool and added to the candidates.

\begin{algorithm}[t]
\caption{Joint AutoDAN attack $\mathcal{A}_{\mathcal{I}}(x) = \textsc{JointAutoDAN}(x, \mathcal{I})$}
\label{alg:joint-autodan}
\begin{algorithmic}[1]
\Require instance $x$; objective models $\mathcal{I} \subseteq [M]$; attacker $g$;
         strategy library $\Sigma$; proposals per anchor $n_{\mathrm{p}}$;
         systemic archive size $n_{\mathrm{sys}}$
\State evaluate $x$ on every $f_j$, $j \in [M]$
\State $a_1, \dots, a_M \gets$ cyclic order of $[M]$
\State $u_0 \gets g\bigl(x, f_{a_1}(x)\bigr)$; evaluate $u_0$ on every $f_j$; \quad $\mathcal{P} \gets \{u_0\}$
\For{$m = 1, \dots, M$}
  \State $u^{\mathrm{sys}} \gets \arg\max_{u \in \mathcal{P}} R_{\mathcal{I}}(u)$
  \State $u^{\mathrm{spec}} \gets \arg\max_{u \in \mathcal{P}} \bigl(b_{a_m}(u), \sigma_{a_m}(u), R_{\mathcal{I}}(u)\bigr)$
  \State $\sigma_1, \sigma_2, \dots \gets$ strategies in $\Sigma$ ranked by similarity to $\bigl(x, u^{\mathrm{sys}}, f_{a_m}(u^{\mathrm{sys}})\bigr)$
  \For{$i = 1, \dots, n_{\mathrm{p}}$}
    \State $\bar{u} \gets u^{\mathrm{sys}}$ if $i$ is odd, else $u^{\mathrm{spec}}$
    \State $S \gets \{\sigma_i\}$ if $i < n_{\mathrm{p}}$, else $\{\sigma_1, \sigma_2\}$
    \State $u \gets g\bigl(x, \bar{u}, f_{a_m}(u^{\mathrm{sys}}), b_{a_m}(u^{\mathrm{sys}}), S\bigr)$
    \State evaluate $u$ on every $f_j$, $j \in [M]$; \quad $\mathcal{P} \gets \mathcal{P} \cup \{u\}$
  \EndFor
\EndFor
\State \Return the top $n_{\mathrm{sys}}$ prompts in $\mathcal{P}$ under $R_{\mathcal{I}}$, and the best $u \in \mathcal{P}$ for each $j \in \mathcal{I}$
\end{algorithmic}
\end{algorithm}

We show further hyperparameter details used in joint AutoDAN experiments in \Cref{tab:hp-joint-autodan}.

\begin{table}[t]
  \centering
  \caption{Hyperparameters of joint AutoDAN.}
  \label{tab:hp-joint-autodan}
  \small
  \begin{tabular}{@{}l p{0.52\linewidth}@{}}
    \toprule
    Hyperparameter & Value \\
    \midrule
    Attacker model $g$ & Qwen3.8-27B \\
    Sampling temperature & 0.8 \\
    Top-$p$ & 0.95 \\
    Maximum new tokens & 768 \\
    Resampling attempts & 3 \\
    Proposals per anchor $n_{\mathrm{p}}$ & 4 \\
    Systemic archive size $n_{\mathrm{sys}}$ & 2 \\
    Samples & 15 HarmBench behaviors \\
    Judge & HarmBench-Llama-2-13b-cls \\
    \bottomrule
  \end{tabular}
\end{table}

\subsection{Further Experimental Details}

\subsubsection{MNIST}\label{app:mnist}

For MNIST experiments, the base model pools are trained once, their configurations persisted and then reused across all experiments. All models, irrespective of pool, are gradient based and thus trained via standard gradient descent, using Adam, specifically. All models are trained to convergence via an early stopping mechanism, triggered when a held out validation loss starts to plateau (with a patience of 10 train steps). After early stopping is triggered, we recover the best performing epoch checkpoint.

Further details on the MNIST model pool composition can be found in \Cref{tab:mnist-pools}.

\begin{table}[t]
  \centering
  \caption{MNIST model pools. Each pool has 20 models, and each diversity level keeps what varies in the level below and adds one more factor. Every model is trained with early stopping and the best-validation weights are kept.}
  \label{tab:mnist-pools}
  \small
  \begin{tabular}{@{}l
                  >{\raggedright\arraybackslash}p{0.25\linewidth}
                  >{\raggedright\arraybackslash}p{0.38\linewidth}
                  >{\raggedright\arraybackslash}p{0.18\linewidth}@{}}
    \toprule
    Pool & Architecture & Training data and regularisation & Varies across models \\
    \midrule
    Low
      & MLPs, one hidden layer (256 units)
      & Shared training splits in fixed order, no augmentation or regularisation
      & Initialisation seed \\
    \addlinespace
    Medium
      & MLPs, one hidden layer (256 units)
      & diverse resampling of training split, random affine
        augmentation, weight decay
      & Seed, training regime \\
    \addlinespace
    High
      & 6 two-layer CNNs, 6 three-layer CNNs, 4 one-hidden-layer MLPs,
        4 \textit{diverse} CNN variants
      & Same as Medium
      & Seed, training regime, architecture family \\
    \bottomrule
  \end{tabular}
\end{table}

\subsubsection{Large Language Models}\label{app:llm_pools}

\begin{figure*}[t]
    \centering
    \includegraphics[width=\textwidth]{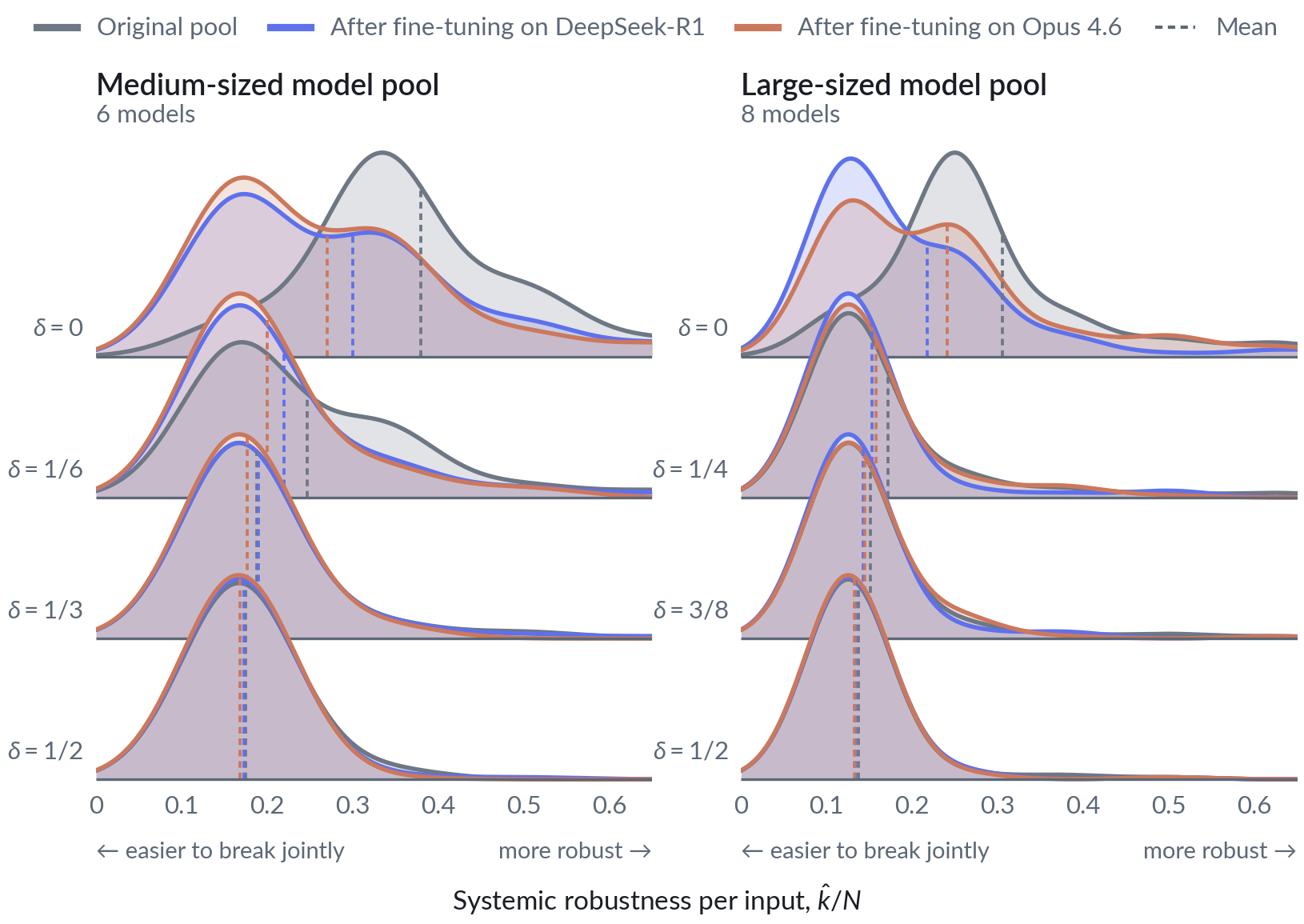}
    \vspace{-10pt}
    \caption{\textbf{Distillation lowers systemic robustness under AdvLLM + BEAST.} Left: the distribution of singleton-completed $\hat{k}/M$ at varying risk tolerances $\delta$ for the medium-sized pool ($M=6$), comparing the original pool, the pool after DeepSeek-R1 distillation, and the pool after subsequent Opus 4.6 distillation.
    Right: same as left but for the large-sized pool ($M=8$). Lower values mean
    that fewer suffixes cover the pool. Dashed lines denote means.}
    \label{fig:advllm-deepseek_opus_distill}
\end{figure*}

\begin{figure*}[t]
    \centering
    \includegraphics[width=\textwidth]{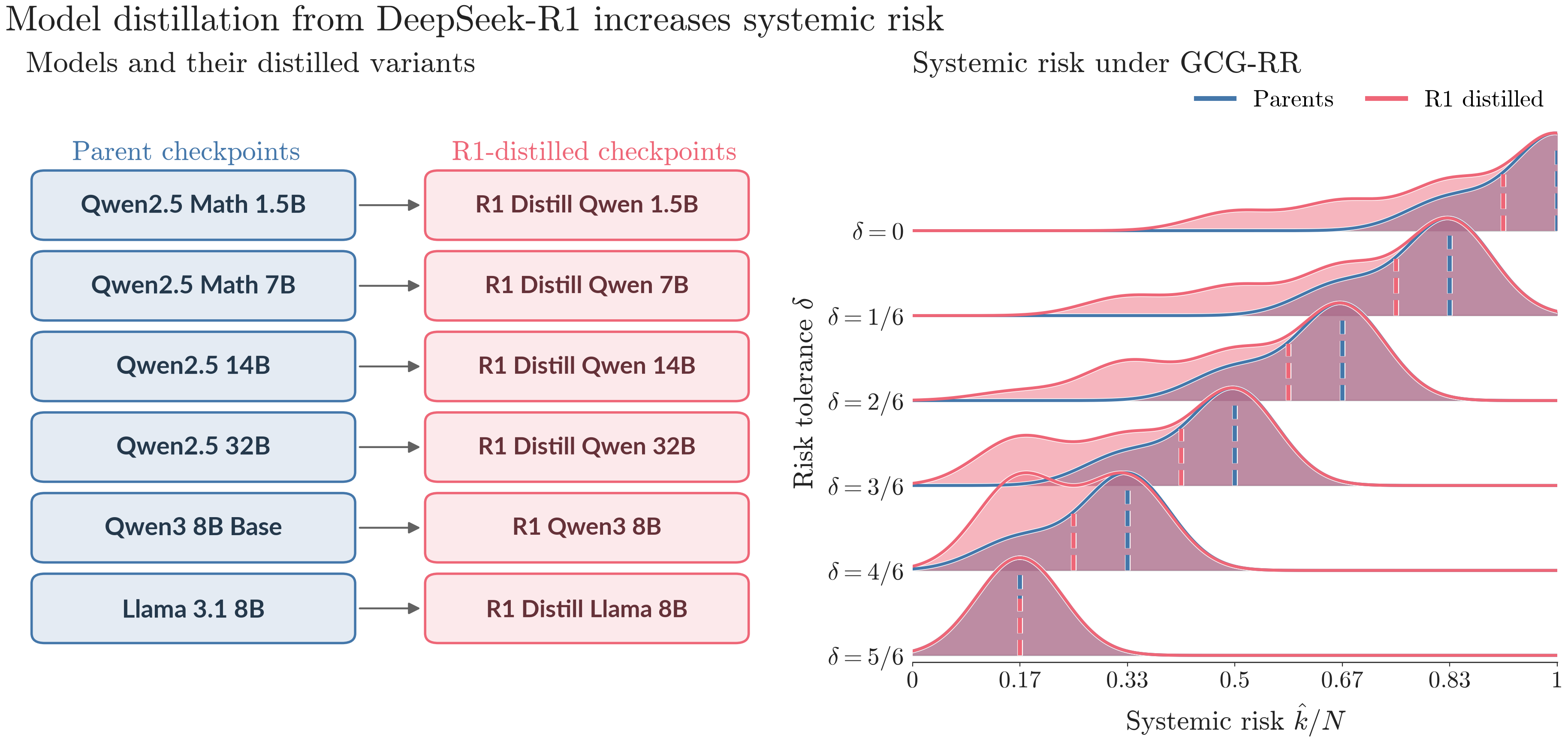}
    \vspace{-10pt}
    \caption{\textbf{Model distillation from DeepSeek-R1 increases raw
    systemic risk under GCG.} Left: the six parent--descendant lineages.
    Right: the distribution of singleton-completed $\hat{k}/6$  at each risk tolerance $\delta$. Lower values mean
    that fewer suffixes cover the pool. Dashed lines denote medians.}
    \label{fig:deepseek_lineage}
\end{figure*}
\vspace{-10pt}

\begin{figure*}[t]
    \centering
    \includegraphics[width=\textwidth]{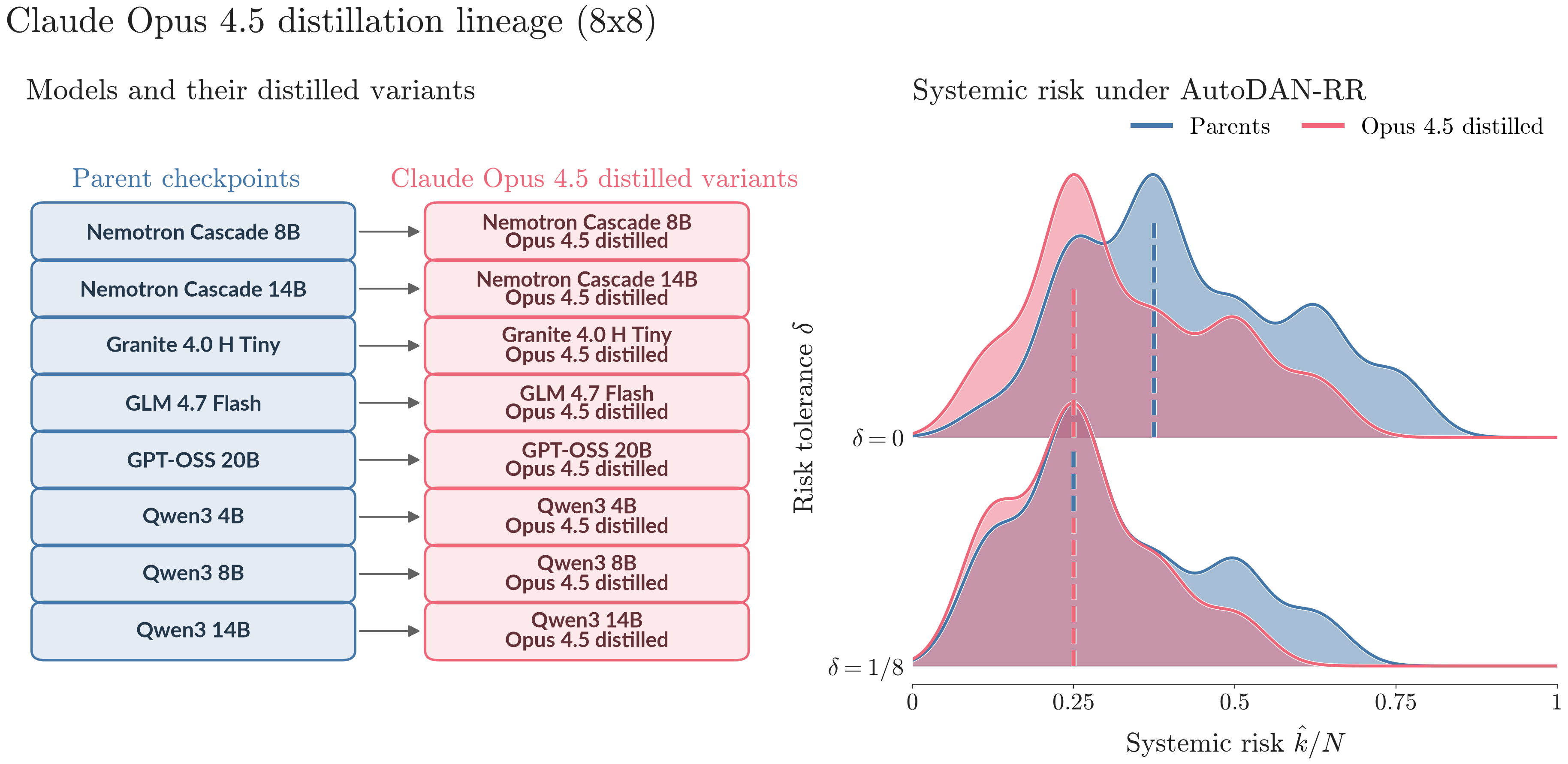}
    \vspace{-20pt}
    \caption{\small \looseness=-1 \textbf{Systemic risk of a Claude Opus 4.5 distilled lineage.}
    Left: eight parent checkpoints and their Opus-attributed descendants.
    Right: singleton-completed systemic risk under AutoDAN-RR. The plotted
    worst-case sets were selected separately for the parents ($25$ behaviors)
    and descendants ($22$ behaviors).}
    \label{fig:opus_lineage}
    \vspace{-20pt}
\end{figure*}

\begin{table}[t]
  \centering
  \caption{Mean systemic robustness per input, $\hat{k}/M$, over the 100 HarmBench behaviours
           (AmpleGCG bank + BEAST, thinking off), with the relative change from the original pool in
           parentheses. Lower values mean the pool is easier to break jointly.}
  \label{tab:khat-means}
  \begin{tabular}{@{}lcccc@{}}
    \toprule
    \multicolumn{5}{@{}l}{\textit{Medium-sized model pool} ($M = 6$)} \\
    & $\delta = 0$ & $\delta = 1/6$ & $\delta = 1/3$ & $\delta = 1/2$ \\
    \midrule
    Original pool                    & 0.543 & 0.385 & 0.270 & 0.203 \\
    DeepSeek-R1 & 0.512 ($-5.8\%$)  & 0.363 ($-5.6\%$)  & 0.267 ($-1.2\%$)  & 0.202 ($-0.8\%$)  \\
    R1 $\rightarrow$ Opus 4.6    & 0.288 ($-46.9\%$) & 0.208 ($-45.9\%$) & 0.180 ($-33.3\%$) & 0.168 ($-17.2\%$) \\
    \midrule
    \multicolumn{5}{@{}l}{\textit{Large-sized model pool} ($M = 8$)} \\
    & $\delta = 0$ & $\delta = 1/4$ & $\delta = 3/8$ & $\delta = 1/2$ \\
    \midrule
    Original pool                    & 0.676 & 0.438 & 0.340 & 0.255 \\
    DeepSeek-R1 & 0.509 ($-24.8\%$) & 0.303 ($-30.9\%$) & 0.238 ($-30.1\%$) & 0.186 ($-27.0\%$) \\
    R1 $\rightarrow$ Opus 4.6    & 0.478 ($-29.4\%$) & 0.279 ($-36.3\%$) & 0.225 ($-33.8\%$) & 0.175 ($-31.4\%$) \\
    \bottomrule
  \end{tabular}
\end{table}

Our experimental evaluation for large language models focuses on two distinct model pools. Both with roughly similar complexity in terms of model size but varying underlying foundation models. The \textit{Opus} pool consists of the following eight models: Nemotron-Cascade-8B-Thinking, Nemotron-Cascade-14B-Thinking, granite-4.0-h-tiny, GLM-4.7-Flash, gpt-oss-20b, Qwen3-4B-Thinking-2507-unsloth-bnb-4bit, Qwen3-8B-unsloth-bnb-4bit, Qwen3-14B. The \textit{Deepseek} pool is comprised of the below six models: Qwen2.5-Math-1.5B, Qwen2.5-Math-7B, Qwen2.5-14B, Qwen2.5-32B, Qwen3-8B-Base, Llama-3.1-8B. See \Cref{tab:pool-characteristics} for an overview. As the pool names suggest, the former is distilled on Opus 4.5 and the latter on DeepSeek-R1 to evaluate the impact of distillation on systemic robustness.

\begin{table}[t]
  \centering
  \caption{LLM pools. Parameters are the
    minimum--maximum over the pool, with the median in parentheses.}
  \label{tab:pool-characteristics}
  \footnotesize
  \setlength{\tabcolsep}{4pt}
  \begin{tabular}{@{}l c c
                  >{\raggedright\arraybackslash}p{0.5\linewidth}@{}}
    \toprule
    Pool & $M$ & Params (B) & Post-training \\
    \midrule
    Opus parents     & 8 & 4.0--31.2 (11.5)
                     & 8 instruction- or reasoning-tuned \\
    Opus descendants & 8 & 4.0--31.2 (11.5)
                     & 8 distilled from reasoning traces attributed to
                       Claude Opus~4.5 \\
    \midrule
    DeepSeek-R1 parents     & 6 & 1.5--32.8 (8.1)
                     & 6 base, not safety-aligned \\
    DeepSeek-R1 descendants & 6 & 1.8--32.8 (8.1)
                     & 6 distilled from DeepSeek-R1\\
    \bottomrule
  \end{tabular}
\end{table}

\end{document}